\documentclass[authorversion,sigconf]{acmart}
\acmConference[Technology Overviews]{Technology Overview Series}{Jan. 23, 2018}{DFINITY Stiftung}
\usepackage{amssymb} %includes amsfonts
\usepackage{amsmath,amsthm}
\usepackage{algorithm,algpseudocode}
\usepackage{enumerate}
\usepackage[font={small,it}]{caption}
\usepackage{nicefrac}
\usepackage{multirow}
\usepackage{pgf,tikz}
\usetikzlibrary{decorations.pathreplacing,arrows}
\usetikzlibrary{shapes.multipart}

\numberwithin{equation}{section}
\newtheorem{assumption}{Assumption}

\newcommand{\ubar}[1]{\text{\b{$#1$}}}
\newcommand{\otau}{\bar{\tau}}
\newcommand{\utau}{\ubar{\tau}}

\newcommand{\eg}{e.g.}

\newcommand{\BB}{\{0,1\}}
\newcommand{\NN}{\mathbb{N}}
\newcommand{\GG}{\mathbb{G}}

\newcommand{\calN}{\mathcal{N}}
\newcommand{\calB}{\mathcal{B}}
\newcommand{\dfn}{\textsc{Dfinity}}
\newcommand{\hash}{H}
\newcommand{\grp}{G}
\newcommand{\group}{\mathsf{Group}}
\newcommand{\prng}{\mathsf{PRG}}
\newcommand{\perm}{\mathsf{Perm}}
\newcommand{\sign}{\mathsf{Sign}}
\newcommand{\verify}{\mathsf{Verify}}

\newcommand{\recover}{\mathsf{Recover}}
\newcommand{\pk}{\mathsf{pk}}
\newcommand{\sk}{\mathsf{sk}}

\newcommand{\btime}{\mathsf{BlockTime}}

\newcommand{\rnd}{\xi}

\DeclareMathOperator{\CDFhg}{CDF_{hg}}
\DeclareMathOperator{\CDFbinom}{CDF_{binom}}
\DeclareMathOperator{\Prob}{Prob}
\DeclareMathOperator{\rank}{rk}
\DeclareMathOperator{\prev}{prv}
\DeclareMathOperator{\nota}{nt}
\DeclareMathOperator{\round}{rd}
\DeclareMathOperator{\data}{dat}
\DeclareMathOperator{\owner}{own}

\DeclareMathOperator{\head}{head}
\DeclareMathOperator{\len}{len}
\DeclareMathOperator{\gen}{gen}

\DeclareMathOperator{\wt}{wt}

\begin{document}
\title[DFINITY Consensus]{DFINITY Technology Overview Series\\Consensus System}
\subtitle{Rev.1}
\date{\today}
\author{Timo Hanke, Mahnush Movahedi and Dominic Williams}

\begin{abstract}
The \dfn\ blockchain computer provides a secure, performant and flexible consensus mechanism.
While first defined for a permissioned participation model,
the consensus mechanism itself can be paired with any method of Sybil resistance
(e.g.\ proof-of-work or proof-of-stake)
to create an open participation model.
\dfn's greatest strength is unfolded in the most challenging proof-of-stake case.

At its core,
\dfn\ contains a decentralized randomness beacon
which acts as a verifiable random function (VRF)
that produces a stream of outputs over time.
The novel technique behind the beacon relies on the existence of a
unique-deterministic, non-interactive, DKG-friendly threshold signatures scheme.
The only known examples of such a scheme are pairing-based
and derived from BLS \cite{BLS,libert2016born}.

The \dfn\ blockchain is layered on top of the \dfn\ beacon
and uses the beacon as its source of randomness for leader selection and leader ranking.
A "weight" is attributed to a chain based on the ranks of the leaders who propose the blocks in the chain,
and that weight is used to select between competing chains.
The \dfn\ blockchain is further hardened by a notarization process
which dramatically improves the time to finality
and eliminates the nothing-at-stake and selfish mining attacks.

\dfn's consensus algorithm is made to scale through continuous quorum selections driven by the random beacon.
In practice,
\dfn\ achieves block times of a few seconds and transaction finality after only two confirmations.
The system gracefully handles temporary losses of network synchrony including network splits,
while it is provably secure under synchrony.
\end{abstract}

\maketitle
\section{Prologue}
DFINITY is a decentralized network design
whose protocols generate a reliable "virtual blockchain computer" running on top of a peer-to-peer network
upon which software can be installed and can operate in the tamperproof mode of smart contracts.
The goal is for the virtual computer to finalize computations quickly
(using short block times and by requiring only a small number of blocks as "confirmations"),
to provide predictable performance (by keeping the time between confirmations approximately constant),
and for computational and storage capacity to scale up without bounds 
as demand for its services increases 
(using novel validation mechanisms and sharding systems discussed in our other papers).
The protocols
must be secure against an adversary controlling less than a certain critical proportion of its nodes,
must generate cryptographic randomness (which is required by advanced decentralized applications)
and must maintain a decentralized nature as it grows in size to millions of nodes.

\dfn\ will be introduced in a series of technology overviews, each highlighting an independent innovation in \dfn\ 
such as the consensus backbone, smart contract language, virtual machine, concurrent contract execution model, daemon contracts, peer-to-peer networks and secure broadcast, governance mechanism and scaling techniques.
The present document will focus on the consensus backbone and cryptographic randomness.

\dfn\ has an unbiasable, verifiable random function (VRF) built-in at the core of its protocol.
The VRF not only drives the consensus,
it will also be the foundation for scaling techniques such as sharding, validation towers, etc.
Moreover, the VRF produced by the consensus layer is available to the application layer,
i.e., to the smart contracts and virtual machine.
In this way, the consensus backbone is intertwined with many of the other topics.

\section{Introduction}\label{sec:intro}

\dfn's consensus mechanism has four layers as depicted in Fig.~\ref{fig:layers}.
The first layer provides registered and Sybil-resistant client identities.
On the second layer is a decentralized random beacon.
On the third layer is a blockchain that is driven by the random beacon
through a probabilistic mechanism for leader ranking.
On the fourth layer is a decentralized notary 
that provides timestamping and publication guarantees,
and is ultimately responsible for near-instant finality.
\dfn's consensus layers and other key aspects of the consensus mechanism can be summarized in the following main categories.
\begin{figure}[ht]
	\centering
	\includegraphics[width=26em]{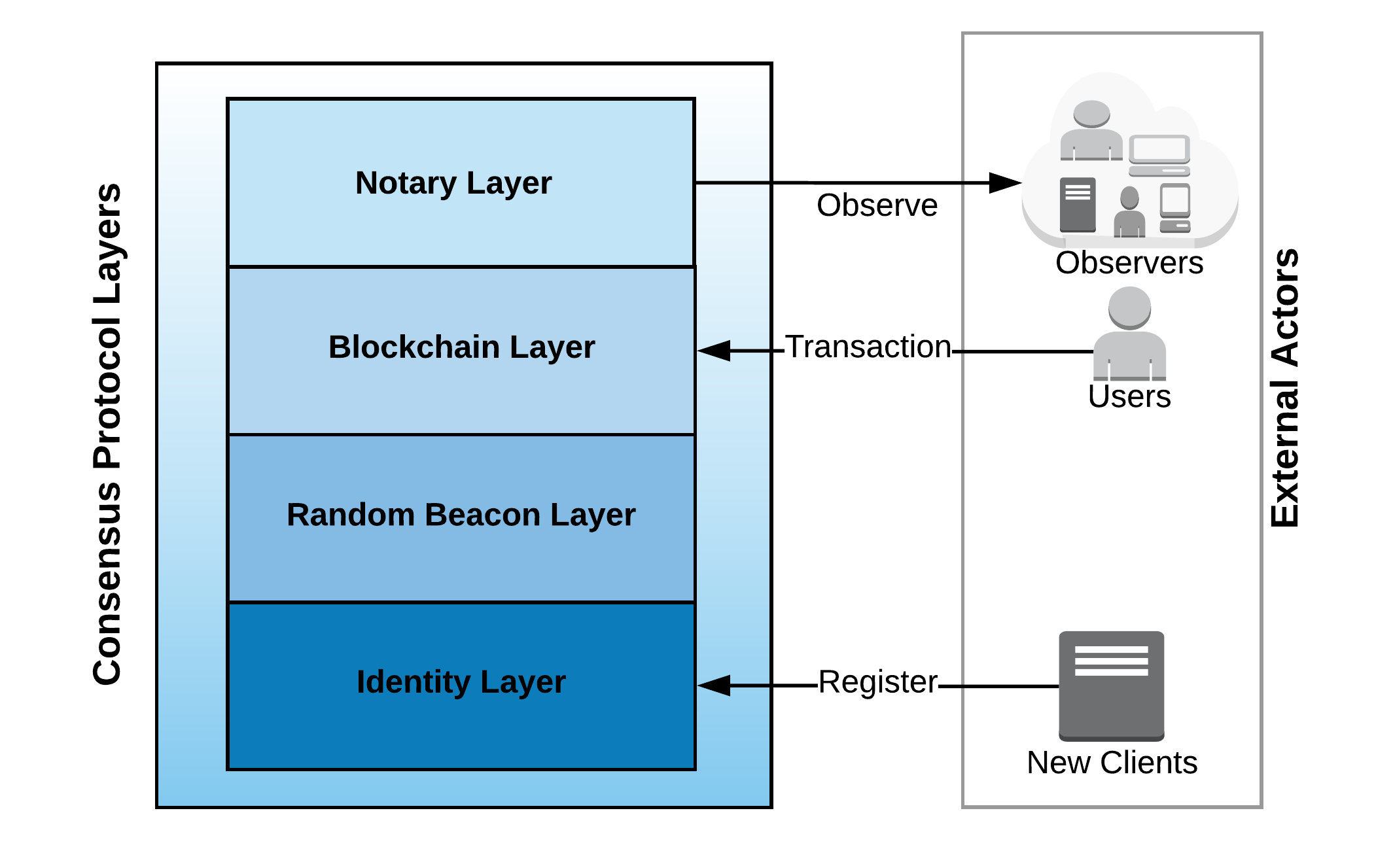}
	\vspace{-0.5em}\caption{\dfn's consensus mechanism layers.
		1.~ Identity layer: provides a registry of all clients.
		2.~ Random Beacon layer: provides the source of randomness (VRF) for all higher layers including applications (smart contracts).
		3.~ Blockchain layer: builds a blockchain from validated transactions via the Probabilistic Slot Protocol driven by the random beacon.
		4.~ Notarization layer: provides fast finality guarantees to clients and external observers.
	}
	\label{fig:layers}	\vspace{-1em}
\end{figure}

\subsubsection*{1st layer: Identities and Registry}
The active participants in the \dfn\ network are called {\em clients}.
All clients in \dfn\ are registered,
i.e., have permanent, pseudonymous identities.
The registration of clients has advantages over the typical proof-of-work blockchains
where it is impossible to link different blocks to the same miner.
For example,
if registration requires a security deposit,
a misbehaving client would lose its entire deposit,
whereas a miner in a typical proof-of-work blockchain
would only forego the block reward during the time of misbehavior.
As a result,
the penalty for misbehavior can be magnitudes larger for registered identities
than it can be for unregistered identities.
This is particularly important as blockchains can track unbounded external value that exceeds the value of the native token itself. 
Moreover,
\dfn\ supports open membership by providing a protocol to register new clients
via a stake deposit with a lock-up period. % (see \S~\ref{sec:open})
This is the responsibility of the first layer.

\subsubsection*{2nd layer: Random Beacon} 
The random beacon in the second layer is an unbiasable, verifiable random function (VRF) that is produced jointly by registered clients.
Each random output of the VRF is unpredictable by anyone until just before it becomes available to everyone.
This is a key technology of the \dfn\ system
which relies on a threshold signature scheme with the properties of uniqueness and non-interactivity.
The BLS signature scheme is the only practical\footnote{RSA-based alternatives exist but suffer from an impracticality of setting up the threshold keys without a trusted dealer.}
scheme that can provide these features,
and \dfn\ has a particularly optimized implementation of BLS built in~\cite{beuchat2010high,blsrepo}.
Using a threshold mechanism for randomness creation solves the fundamental "last actor" problem.
Any decentralized protocol for creating public randomness without a threshold mechanism suffers from the problem that the last actor in that protocol
knows the next random value and can decide to abort the protocol.

\subsubsection*{3rd layer: Blockchain and fork resolution}
The third layer deploys the "probabilistic slot protocol" (PSP).
This protocol ranks the clients for each height of the chain,
in an order that is derived deterministically from the unbiased output of the random beacon for that height.
A weight is then assigned to block proposals based on the proposer's rank
such that blocks from clients at the top of the list receive a higher weight.
Forks are resolved by giving favor to the "heaviest" chain in terms of accumulated block weight --
quite similar to how traditional proof-of-work consensus
is based on the highest accumulated amount of work.
The first advantage of the PSP protocol is that the ranking is available instantaneously,
which allows for a predictable, constant block time.
The second advantage is that there is always a single highest-ranked client
which allows for a homogenous network bandwidth utilization.
Instead, a race between clients would favor a usage in bursts.

\subsubsection*{4th layer: Notarization and near-instant finality.}
Finality of a given transaction means a system-wide consensus that a given transaction has been irreversibly executed.
While most distributed systems require rapid transaction finality, 
existing block\-chain techniques are unable to provide it.
\dfn\ deploys the novel technique of block notarization in its fourth layer to speed up finality.
A notarization is a threshold signature under a block 
created jointly by registered clients.
Only notarized blocks can be included in a chain.
Of all the block candidates that are presented to a client for notarization,
the client only notarizes the highest-ranked one with respect to a publicly verifiable ranking algorithm
driven by the random beacon.
It is important to emphasize that notarization is not consensus
because it is possible, due to adverse timing,
for more than one block to get notarized at a given height.
This is explicitly tolerated
and an important difference to other proof-of-stake proposals
that apply full Byzantine agreement at every block.
\dfn\ achieves its high speed and short block times exactly because notarization is {\em not} full consensus.
However, notarization can be seen as optimistic consensus
because it will frequently be the case that only one block gets notarized.
Whether this is the case can be detected after one subsequent block 
plus a relay time (cf.\ Theorem~\ref{thm:main}).
Hence, whenever the broadcast network functions normally
a transaction is final in the \dfn\ consensus after two notarized confirmations
plus a network traversal time.

We like to emphasize that a notarization in \dfn\ is not primarily a validity guarantee
but rather a timestamp plus a proof of publication.
The notarization step makes it impossible for the adversary to build and sustain a chain of linked, notarized blocks
in secret.
For this reason, \dfn\ does not suffer from the selfish mining attack \cite{eyal2014majority}
or the nothing-at-stake problem.

\subsubsection*{Threshold Relay and Network Scalability.}
\dfn's consensus is designed to operate on a network of millions of clients.
To enable scalability to this extent,
the random beacon and notarization protocols
are designed such that they can be safely and efficiently delegated to a committee.
A {\em committee} is a randomly sampled subset of all registered clients
that deploys a threshold mechanism (for safety) that is moreover non-interactive (for efficiency).

In \dfn, the active committee changes regularly.
After having temporarily executed the protocol on behalf of all clients,
the committee relays the execution to another pre-configured committee.
We call this technique "Threshold Relay" in \dfn.

\subsubsection*{Consistency vs availability}
It is worth noting that network splits are 
implicitly detectable by \dfn\ and are handled conservatively.
This is a consequence of the random sampling of committees.
If the network splits in two halves of more or less the same size,
this will automatically cause the random beacon to pause within a few blocks
so that none of the sides can continue. 
The random beacon will automatically resume once the network reconnects.
If the network splits in a way that one component is significantly larger than half of the network,
the protocol may continue in that one large component but will pause in all other components.

Network splits can not only occur when the communication is interrupted.
Another important and even more realistic case is 
when there are multiple implementations of the \dfn\ client 
and they disagree due to the exposures of a bug.
\dfn\ handles this case gracefully.
If there are two clients in evenly widespread use and they start to disagree, then both clients will pause.
If there are many evenly spread clients and one starts to disagree from all the others,
then the network will likely continue and only the isolated client will pause.
This is exactly the desired behavior in the given scenarios.
Other blockchains do not handle this case well and the occurrence of such an event poses a real threat to them.
The reason is that these chains put too much emphasis on availability rather than consistency.

\subsubsection*{Paper organization.}
\S~\ref{sec:overview} presents a high-level view of the protocol.
\S~\ref{sec:model} specifies our system, communication and threat models
and introduces relevant notations.
\S~\ref{sec:psp}-\ref{sec:drb}
describe the probabilistic slot protocol and random beacon protocol in detail.
\S~\ref{sec:tr} introduces the Threshold Relay technique
which allows the protocols to be safely executed by pre-configured committees rather than by all replicas.
\S~\ref{sec:open} describes the open participation model
which allows members to join and leave the protocol over time.
Finally, \S~\ref{sec:security} provides the security and correctness proofs for the \dfn\ protocol.

\section{A high-level view of the Consensus Protocol}\label{sec:overview}
\subsubsection*{Roles}
The \dfn\ peer-to-peer network consists of clients
connected by a broadcast network over which they can send messages to everyone.
Clients fulfill three active functions:
(a) participate in the decentralized random beacon,
(b) participate in the decentralized notary,
(c) propose blocks.
Clients also observe blocks and build their own view of the finalized chain.

\subsubsection*{Committees and Threshold Relay}
To improve scalability,
the random beacon and notary are run by a {\em committee}.
In a small scale network the committee can be the set of all clients.
In a large scale network,
the committee is smaller than the set of all clients
and changes from round to round (i.e., from block to block).
The random beacon output in one round chooses the committee for the next round
according to the {\em threshold relay} technique described in \S~\ref{sec:tr}.
The committee size is configured based on a failure probability calculation (see \S~\ref{ssec:grpsize}).

\subsubsection*{Block ranking}
If we abstract away the decentralized aspect of the random beacon and notary
then the consensus protocol is depicted in Fig.~\ref{fig:overview}.
The protocol proceeds in rounds such that there is a one-to-one correspondence between the round number and the position (called {\em height}) in the chain.
At the beginning of round $r$, the randomness beacon produces a fresh, verifiable random value and broadcasts it to the network (Fig.~\ref{fig:overview}, step 1).
The random beacon output for round $r$, denoted by $\rnd_r$,
determines a priority ranking of all registered clients.
Any client can propose a block %for the current round
but a higher priority level of the client means a higher chance
that the block will get notarized
and that block makers in the subsequent round will build on top of it.

\subsubsection*{Notarization}
Once a client sees a valid
$\rnd_r$, it pools transactions collected from users into a block candidate and sends it to the notary (Fig.~\ref{fig:overview}, step 2).
The notary waits for a specific constant time ($\btime$) to receive the proposed blocks.
Then, the notary runs the ranking mechanism based on the random beacon,
chooses the highest-ranked block,
and signs and broadcasts it (Fig.~\ref{fig:overview}, step 3).
As soon as clients receive a notarized block,
they use it to extend their copies of the blockchain thereby
ending round $r$ in their respective views.
Finally, the random beacon
broadcasts $\rnd_{r+1}$ which marks the beginning of a new round.

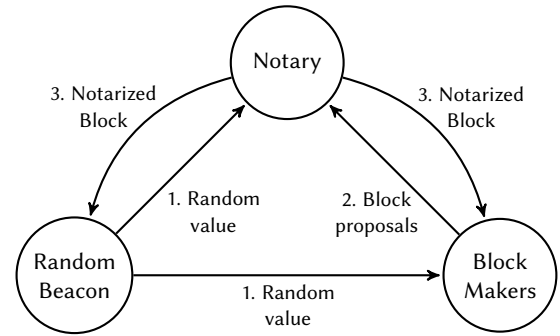
\begin{figure}
	\centering
\begin{tikzpicture}[->,>=stealth',shorten >=1pt,auto,node distance=4cm,
                    thick,main node/.style={circle,draw,font=\sffamily,minimum size=1.5cm},
                    every text node part/.style={align=center}]

  \node[main node] (3) {Notary};
  \node[main node] (1) [below left of=3] {Random \\ Beacon};
  \node[main node] (2) [below right of=3] {Block \\ Makers};

  \path[every node/.style={font=\sffamily\small}]
    (1) edge node [below, pos=1/2] {1. Random \\ value} (2)
        edge node [below,pos=0.4,xshift=6mm] {1. Random \\ value} (3)
    (2) edge node [below,pos=0.4,xshift=-4mm] {2. Block \\ proposals} (3)
    (3) edge [bend left] node [right,pos=1/3] {3. Notarized \\ Block} (2) 
    (3) edge [bend right] node [left,pos=1/3] {3. Notarized \\ Block} (1); 
\end{tikzpicture}
\caption{High-level overview on the system components.
		1.~The random beacon in round $r$ produces random output $\xi_r$.
		2.~The block maker(s) selected deterministically by $\xi_r$ propose block(s) for round $r$.
		3.~The decentralized notary notarizes the block(s) of the preferred block maker(s).
		4.~The random beacon advances to round $r+1$ upon seeing a notarized block from round $r$.
	}
	\label{fig:overview}	\vspace{-1em}
\end{figure}

\subsubsection*{Decentralized Random Beacon}
The random beacon protocol is completely decentralized and operated by all clients in the committee together.
Nevertheless from the outside
(i.e., looking only at the outputs produced and the timing of the outputs),
the beacon behaves like a trusted third party.
We emphasize that the committee
does not need to run a Byzantine agreement protocol
for every output that the beacon produces.
Instead, agreement on each of the beacon's output is automatic
because of the uniqueness property of our threshold signature scheme.
This explains how the random beacon can run at such high speed,
and thereby the \dfn\ blockchain can achieve such a low block time.

\subsubsection*{Decentralized Notary}
As was the case for the random beacon,
the notary is completely decentralized and operated by all clients in the committee together
and its behavior as a whole can be equated to a trusted third party.
However, unlike the random beacon,
the notary seeks to agree on live input -- a block --
rather than on a pseudo-random number.
There is no "magic" cryptography available for this,
so a full Byzantine agreement protocol would be the only option.
But instead of doing that,
the \dfn\ notary merely runs an optimistic protocol
which achieves consensus "under normal operation"
though may sometimes notarize more than one block per round.
If this happens, \dfn's chain ranking algorithm will resolve the fork
and finality can be achieved in a subsequent normal round.
The optimistic protocol is non-interactive and fast,
hence the notary can run at the same speed as the random beacon.

\section{Models and Preliminaries} \label{sec:model}

\subsection{System Model}
\subsubsection{Replicas}
From now on we speak of clients as {\em replicas} and label them $1,2,\ldots\in\NN$.
Let $U$ be the finite set of labels of all replicas, called the {\em universe}.
Each replica $i\in U$ has a public/private key pair $(\pk_i, \sk_i)$.
We assume the set $\{\pk_i\,|\,i\in U\}$ is known and agreed upon\footnote{In practice,
the agreement is achieved by registering all replicas on the blockchain with their public key.}
 by all $i\in U$.

\subsubsection{Authentication}
Each protocol message is signed by the replica that issues the message.
The replicas only accept and act upon a message
if the message is signed by one of the $\sk_i, i\in U$.

\subsubsection{Groups}
At any given time,
some or all $i\in U$ are arranged into one or more subsets
$G_1, G_2, \ldots\subseteq U$ called \emph{groups},
of which a single one,
the \emph{committee},
is active to drive progress and ensure consensus.
We assume all groups $G_j$ have the same size $n$.
The number $n$ is a system parameter called the {\em group size}.

\subsubsection{Synchrony}
For the practical use of \dfn\ we assume a {\em semi-synchronous} network
by which we mean that the network traversal time can be modeled by a random variable $Y$
whose probability distribution is known.
The \dfn\ protocol then chooses two system-wide timeout constants $\btime$ and $T$
based on the distribution of $Y$
%and taking into account certain trade-offs between speed
and the security parameter of the system.
In the formal security analysis in \S~\ref{sec:security}
we give proofs for the synchronous case in which an upper bound $\Delta$ for $Y$ is known.

The two constants are responsible for liveness ($\btime$) and safety ($T$) of the system, respectively.
Timeout clocks are triggered based on local {\em events},
i.e.\ received messages.
The protocol does not depend on a global time
nor does it assume synchronized clocks between the replicas.

The system evolves in \emph{rounds}.
Replicas advance to the next round based on events.
The rounds are not expected to be in sync across different replicas.

\subsection{Threat Model}
\subsubsection{Byzantine replicas}
A replica that faithfully follows the protocol is called {\em honest} and all other replicas are called {\em Byzantine}.
A Byzantine $i\in U$ may behave arbitrarily, \eg,
it may refuse to participate in the protocol or it may collude with others to 
perform a coordinated attack the system.

\subsubsection{Adversarial strength}
For any $G\subseteq U$ let $f(G)$ denote the number of Byzantine replicas in $G$.
\begin{assumption}\label{ass:1}
There is $\beta>2$ such that
\begin{equation}\label{eq:betaf}
|U|>\beta f(U).
\end{equation}
\end{assumption}
The value $1/\beta$ is called the {\em adversarial strength}.
In practice, 
Assumption \ref{ass:1} is achieved through economic incentives in conjunction with a form of Sybil resistance.\footnote{
Sybil resistance is achieved for example by requiring a stake deposit for each replica.
Then Assumption \ref{ass:1} translates to the assumption
that at least a $1/\beta$ fraction of stake deposits
were made by honest participants.
}

\subsubsection{Honest groups}
Let $n$ be the group size. 
Then a group $G$ is called {\em honest} if
\begin{equation}\label{eq:2fG}
n>2f(G).
\end{equation}
The protocols described in \S~\ref{sec:psp}-\ref{sec:drb} rely on 
\begin{assumption}\label{ass:2}
Each group $G$ used in the system is honest.
\end{assumption}

\subsubsection{Random samples}\label{ssec:grpsize}
Given Assumption \ref{ass:1}, the universe $U$ itself is honest.
Each group $G\subseteq U$ used in the system is a random sample of size $n$ drawn from $U$.
Given $n$, the probability $\Prob[\textrm{$G$ honest}]$
can be calculated as follows:

\begin{proposition}\label{prop:probhg}
Let $\CDFhg(x,n,M,N)$ denote the cumulative distribution function of the hypergeometric probability distribution
where $N$ is the population size, $M$ is the number of successes\footnote{In our application $M$ is the number of Byzantine replicas in $U$.} in the population,
$n$ is the sample size and $x$ is the maximum number of successes allowed per sample.
Then
\begin{equation}\label{eq:cdfhg}
 \Prob[\textrm{$G$ honest}] = \CDFhg(\lceil n/2\rceil-1,n,\lfloor |U|/\beta\rfloor,|U|).
\end{equation}
\end{proposition}
Given an acceptable {\em failure probability} $\rho$,
we can solve \eqref{eq:cdfhg} for the minimal group size $n=n(\beta,\rho,|U|)$
such that
\[ \Prob[\textrm{$G$ honest}] > 1-\rho \]
for each random sample $G\subseteq U$ with $|G|=n$.
The result for the example value $|U|=10^4$ and different values for $\rho$ and $\beta$ are shown in Figure \ref{fig:n} below.
\begin{figure}[h]
\begin{tabular}{c||c|c|c}
\multirow{2}{*}{$-\log_2\rho$} & \multicolumn{3}{c}{$n(\beta,\rho,10^4)$} \\
& $\beta=3$ & $\beta=4$ & $\beta=5$ \\
\hline
40 & 405 & 169 & 111\\
64 & 651 & 277 & 181\\
80 & 811  & 349 & 227\\
128 & 1255 & 555 & 365
\end{tabular}
\caption{Minimal group size for $|U|=10^4$.
Example values for the minimal group size $n(\beta,\rho,10^4)$ for adversarial strength $1/\beta$ and failure probability $\rho$.}
\label{fig:n}
\end{figure}

As the population size increases to infinity the hypergeometric distribution converges to the binomial distribution.
Thus, as $|U|$ increases to infinity we get
\begin{proposition}\label{prop:probbinom}
Let $\CDFbinom(x,n,p)$ denote the cumulative distribution function of the binomial probability distribution
where $p$ is the success probability per draw, $n$ is the sample size and $x$ is the maximum number of successes allowed per sample.
Then
\begin{equation}\label{eq:cdfbinom}
 \Prob[\textrm{$G$ honest}] \geq \CDFbinom(\lceil n/2\rceil-1,n,1/\beta).
\end{equation}
\end{proposition}
Given $\rho$,
we can solve \eqref{eq:cdfbinom} for $n$ 
and get the minimal group size $n(\beta,\rho)$
such that $n(\beta,\rho)\geq n(\beta,\rho,|U|)$ for all values of $|U|$.

% f(n,p,t) := CDF(X <= t) where X is n random samples with a success probability p.
% general rule: f(n,p,t) = 1-f(n,1-p,n-t-1)
% general rule: f(n,p,t+1) = 1-f(n,1-p,n-t)
% for n odd:
%   failure probability = Prob(G not honest) = f(n,1-1/beta,(n-1)/2) = 1-f(n,1/beta,(n-1)/2)
% for n even:
%   failure probability = Prob(G not honest) = f(n,1-1/beta,n/2) = 1-f(n,1/beta,n/2-1)
% for any n:
%   failure probability = Prob(G not honest) = f(n,1-1/beta,floor(n/2)) = 1-f(n,1/beta,ceil(n/2-1))

The result for different values for $\rho$ and $\beta$ are shown in Figure \ref{fig:n2} below.
As one can see, within the range of interest for $\rho$,
the group size is approximately linear in $-\log_2\rho$.
The resulting group sizes are practical for the protocols described in this paper.\footnote{
The main protocols described in this paper are so-called "non-interactive".
Group sizes of 1,000 have been tested in implementations and were proven to be unproblematic.
\dfn\ plans to launch its network with group sizes in the order of 400.}

\begin{figure}[h]
\begin{tabular}{c||c|c|c}
\multirow{2}{*}{$-\log_2\rho$} & \multicolumn{3}{c}{$n(\beta,\rho)$} \\
& $\beta=3$ & $\beta=4$ & $\beta=5$ \\
\hline
40 & 423 & 173 & 111\\
%60 & 653 & 269 & 173\\
64 & 701 & 287 & 185\\
80 & 887  & 363 & 235\\
%120 & 1355  & 555 & 357\\
128 & 1447 & 593 & 383\\
%160 & 1823 & 747 & 481\\
%256 & $\approx$ 2970 & 1207 & 779
\end{tabular}
\caption{Minimal group size for arbitrarily large $U$.
Example values for the minimal group size $n(\beta,\rho)$ for adversarial strength $1/\beta$ and failure probability $\rho$.}
\label{fig:n2}
\end{figure}

\subsubsection{Adaptive adversary}
We assume that the adversary is \emph{mildly
	adaptive}.
This means the adversary may adaptively corrupt groups but
this corruption takes longer than the activity period of the group.

\subsection{Cryptographic Primitives}
\subsubsection{Hash function}
We assume we have a collision-resistant hash function $\hash$
with digests of bit-length $l$ where $l$ matches the security parameter $\kappa$.

\subsubsection{Pseudo-random numbers}
We also assume we have a cryptographically secure pseudo-random number generator $\prng$
which turns a seed $\xi$ into a sequence of values $\prng(\xi,i)$
for $i=0,1,\ldots$.

\subsubsection{Pseudo-random permutations}
The sequence $\prng(\xi,i)$ can be used as input to the Fisher-Yates shuffle
\cite[Algorithm 3.4.2P]{knuth1997art}
to produce a random permutation of $U$.
The result is an bijective map $\{1,\ldots,|U|\}\to U$
which we denote by $\perm_U(\xi)$.

\subsubsection{Diffie-Hellman}
We assume that the adversary is bounded computationally
and that the computational Diffie-Hellman problem is hard
for the elliptic curves with pairings in \cite{beuchat2010high}.

\subsection{\dfn's Block Chain}\label{sec:bc}
We now define formally the concept of a blockchain in \dfn.
\subsubsection{Blocks}
\begin{definition}\label{def:block}
A \emph{block} is either a special {\em genesis block} or a tuple
$(p, r, z, d, o)$ where
$p\in\BB^l$ is the hash reference to the previous block,
$r\in\NN$ is the round number,
$z\in\BB^*$ is the notarization of the previous block,
$d\in\BB^*$ is the data payload ("transactions" and "state"),
$o\in U$ is the creator (or "owner").
A {\em notarization} is a signature on the previous block created by a "notary".
For a block $B=(p, r, z, d, o)$ we define
\[
\begin{aligned}
\prev B:=p, &&
\nota B:=z, &&
\round B:= r, &&
\data B:=d, &&
\owner B:= o. 
\end{aligned}
\]
\end{definition}

We emphasize that a
block contains the notarization $z$ of the previous block in the chain that it references.

\subsubsection{Chains}
\begin{definition}\label{def:chain}
By a {\em chain} $C$ we mean a finite sequence of blocks $(B_0, B_1, \ldots, B_r)$ with
$\round B_i = i$ for all $i$, 
$\prev B_i=\hash(B_{i-1})$ for all $i>0$,
and $\nota B_i$ a valid signature of $B_{i-1}$ for all $i>0$.
The first block $B_0$ is a genesis block.
The last block $B_r$ is called the {\em head} of $C$.
We define
\[
\begin{aligned}
\len C := r+1, && \gen C := B_0, && \head C := B_r.
\end{aligned}
\]
\end{definition}

Since blocks in a chain are linked through cryptographic hashes,
a chain is an authenticated data structure.
A chain is completely determined by its head by virtue of
\begin{proposition}
It is computationally infeasible to produce two chains $C\neq C'$ with $\head C=\head C'$.
\end{proposition}
\begin{definition}
We write $C(B)$ for the uniquely defined chain $C$ with $\head C=B$.
Given two chains $C,C'$ we write $C\leq C'$ if $C$ is a prefix of $C'$.
\end{definition}
Assume from now on that all chains have the same genesis block $B_0$.
\begin{definition}
For any non-empty set $S$ of blocks we denote by $C(S)$ the largest common prefix of all chains $C(B)$ with $B\in S$.
\end{definition}
The chain $C(S)$ is defined because every $C(B), B\in S,$ contains the genesis block.
For any sets of blocks $S,T$ with $S\subseteq T$ we have $C(T)\leq C(S)$.
%This implies:
%\begin{equation}\label{eq:X}
%S\cap T\neq\emptyset \Rightarrow C(S)\leq C(T) \textrm{ or } C(T)\leq C(S).
%\end{equation}
Suppose $\prev S := \{\prev B\,|\, B\in S\setminus\{B_0\}\}\neq\emptyset$.
Then
\begin{equation}\label{eq:prev}
C(\prev S)\leq C(S).
\end{equation}
% In fact if $|S|\geq 2$ then even equality holds.

\section{Probabilistic Slot Protocol and Notarization}\label{sec:psp}

As was explained in \S~\ref{sec:overview},
each protocol round runs through the steps of producing a random beacon output (1),
producing block proposals (2), and producing block notarizations (3).
Since more than one block can get notarized, these steps alone do not provide consensus.
This is where the \emph{probabilistic slot protocol} (PSP) steps in.

Based on {\em block weight}, the PSP allows replicas to decide
which chain to build on when they propose a new block.
Over time, this leads to probabilistic consensus on a chain prefix,
where the probability of finality increases the more "weight" is added to a chain.
This is analogous to proof-of-work chains,
where the probability of finality increases the more "work" is added to a chain.
However, \dfn\ does not stop here and does not rely on this probabilistic type of finality decisions.
PSP is only used to guide the block proposers.
For finality, \dfn\ applies a faster method utilizing a {\em notarization} protocol.

For this section,
we assume the random beacon (which we introduce later in \S~\ref{sec:drb})
is working without failure and provides all replicas with a new, unbiased random value  $\rnd_r$ at the start of each round $r$.
Figure \ref{fig:steps} shows how the protocol alternates between extending the blockchain and extending the random beacon chain
and demonstrates how the random beacon, block proposer and notary
advance in lockstep.

\usetikzlibrary{arrows}

\tikzset{Block/.style={draw, minimum size=22,line width=1, thick, draw=black!80, top color=white,bottom color=black!10}}
\tikzset{NewBlock/.style={loosely dashed,draw,minimum size=22,line width=1, thick, draw=black!80, top color=white,bottom color=black!10}}
\tikzset{Arrow/.style={-stealth, thick,draw=black!80}}
\tikzset{Tail/.style={dashed,ultra thick,line width=1,draw=black!80}}
\begin{figure*}

\begin{tabular*}{\textwidth}{l l l l}

\begin{minipage}[t]{0.24\textwidth}
\centering
\begin{tikzpicture} [scale=0.6,>=stealth',font=\sffamily] 

%No. 1
\node [Block] (v1) at (1.5,4) {$\xi_{r-1}$};
\node [Block] (v2) at (4,4) {$\xi_{r}$};
\draw [Arrow] (v2) edge (v1);
\node (v3) at (0,4) {};
\draw [Tail] (v3) edge (v1);

\node [Block] (v4) at (1.5,2) {$B_{r-1}$};
\node [NewBlock] (v5) at (4,2) {$B_{r}$};
\draw [Arrow] (v5) edge (v4);
\node (v6) at (0,2) {};
\draw [Tail] (v6) edge (v4);
\node at (1.5,1) {$z_{r-1}$};

%Random beacon shares
\node (a1) at (6,2.5) {$\sigma_{B_{r},i_{1}}$};
\node (a2) at (6,1.9) {$\sigma_{B_{r},i_{2}}$};
\node (a3) at (6,1.5) {...};
\node (a4) at (6,1.1) {$\sigma_{B_{r},i_{t}}$};
\draw [Arrow] (a1) edge (v5);
\draw [Arrow] (a2) edge (v5);
\draw [Arrow] (a4) edge (v5);

\end{tikzpicture}
\captionsetup{font=footnotesize}
\caption*{1. In round $r$, a block $B_r$ is being proposed by a block maker that is prioritized by $\xi_r$. $B_r$ references the previous block $B_{r-1}$. Members of the notary committee, which is selected by $\xi_r$, sign off $B_r$.}
\end{minipage}

&

\begin{minipage}[t]{0.20\textwidth}
	\centering
	\begin{tikzpicture} [scale=0.6,>=stealth',font=\sffamily] 

%No. 2
\node [Block] (v1) at (2,4) {$\xi_{r-1}$};
\node [Block] (v2) at (4.5,4) {$\xi_{r}$};
\draw [Arrow] (v2) edge (v1);
\node (v3) at (0.5,4) {};
\draw [Tail] (v3) edge (v1);

\node [Block] (v4) at (2,2) {$B_{r-1}$};
\node [Block] (v5) at (4.5,2) {$B_{r}$};
\draw [Arrow] (v5) edge (v4);
\node (v6) at (0.5,2) {};
\draw [Tail] (v6) edge (v4);
\node at (2,1) {$z_{r-1}$};
\node at (4.5,1) { $z_{r}$};

\end{tikzpicture}
\captionsetup{font=footnotesize}
\caption*{2. $B_r$ has received signature from a majority of replicas and is now notarized by virtue of the aggregated signature $z_r$.
	Every replica enters round $r+1$ upon seeing $z_r$.}
\end{minipage}

&

\begin{minipage}[t]{0.24\textwidth}
\centering
\begin{tikzpicture} [scale=0.6,>=stealth',font=\sffamily] 

%No. 3
\node [Block] (v1) at (2,4) {$\xi_{r-1}$};
\node [Block] (v2) at (4.5,4) {$\xi_{r}$};
\draw [Arrow] (v2) edge (v1);
\node (v3) at (0.5,4) {};
\draw [Tail] (v3) edge (v1);

\node [Block] (v4) at (2,2) {$B_{r-1}$};
\node [Block] (v5) at (4.5,2) {$B_{r}$};
\draw [Arrow] (v5) edge (v4);
\node (v6) at (0.5,2) {};
\draw [Tail] (v6) edge (v4);
\node at (2,1) {$z_{r-1}$};
\node at (4.5,1) { $z_{r}$};

%Notarization
\node (a5) at (6.65,4.5) {$\sigma_{r||\xi_{r},i_{1}}$};
\node (a6) at (6.65,3.9) {$\sigma_{r||\xi_{r},i_{2}}$};
\node (a7) at (6.65,3.5) {...};
\node (a8) at (6.65,3.1) {$\sigma_{r||\xi_{r},i_{t}}$};
\draw [Arrow] (a5) edge (v2);
\draw [Arrow] (a6) edge (v2);
\draw [Arrow] (a8) edge (v2);

\end{tikzpicture}
\captionsetup{font=footnotesize}
\caption*{3. Members of the random beacon committee, which is selected by $\xi_r$, sign the previous randomness $\xi_r$ right after entering round $r+1$.}
\end{minipage}

&

\begin{minipage}[t]{0.24\textwidth}
	\centering
	\begin{tikzpicture} [scale=0.6,>=stealth',font=\sffamily]

%No. 4

\node [Block] (v1) at (2,4) {$\xi_{r-1}$};
\node [Block] (v2) at (4.5,4) {$\xi_{r}$};
\draw [Arrow] (v2) edge (v1);
\node (v3) at (0.5,4) {};
\draw [Tail] (v3) edge (v1);

\node [Block] (v4) at (2,2) {$B_{r-1}$};
\node [Block] (v5) at (4.5,2) {$B_{r}$};
\draw [Arrow] (v5) edge (v4);
\node (v6) at (0.5,2) {};
\draw [Tail] (v6) edge (v4);
\node at (2,1) {$z_{r-1}$};

\node [Block] (v7) at (7,4) {$\xi_{r+1}$};
\draw [Arrow] (v7) edge (v2);
\node at (4.5,1) { $z_{r}$};

\end{tikzpicture}
\captionsetup{font=footnotesize}
\caption*{4. The next random beacon output $\xi_{r+1}$ is formed as a unique threshold signature.
	The cycle continues with step 1 for round $r+1$.}
\end{minipage}
\end{tabular*}
\caption{Alternation between the random beacon chain and the notarized blockchain.}
\label{fig:steps}
\end{figure*}

For the exposition of the present section, however,
the decentralized nature and precise inner workings of the random beacon are irrelevant.
Hence we simply regard the sequence $\rnd_r$ as given without making further assumptions about it.

Regarding the threat model,
we assume \eqref{eq:2fG} for all groups,
as stated in Assumption \ref{ass:2}.
However, for the description and understanding of the notarization protocol
it is sufficient to assume that there is only a single group consisting of the universe of replicas $U$
and that
\begin{equation}\label{eq:2fU}
|U|>2f(U).
\end{equation}
For simplicity of exposition we do adopt this view.
It may then be apparent that the protocol described in this section
can be delegated to any honest committee or sequence of honest committees.

\subsection{Block Rank and Chain Weight}
Based on $\rnd_r$,
the protocol assigns a {\em rank} to each $i\in U$
and the rank of the proposer defines the {\em weight} of a block as follows.

\begin{definition}[Replica Ranking]
The {\em ranking permutation} for round $r$ is defined as $\pi_r := \perm_U(\xi_r)$.
The {\em rank} of $i\in U$ in round $r$ is defined as $\pi_r(i)$.
\end{definition}

\begin{definition}[Block Ranking]
The {\em rank} of a block $B$ is defined as
$\rank B:=\pi_r(\owner B)$ where $r=\round B$.
We say $B$ has {\em higher priority} level than $B'$ if $\rank B<\rank B'$.
\end{definition}
If an adversary equivocates then there will be
multiple blocks for the same round with the same rank.

We assume the protocol has defined a monotonically decreasing function $w$.
In particular, for \dfn\ we instantiate $w$ as $w(x) = 2^{-x}$.
\begin{definition}[Block Weight]
The {\em weight} of a block $B$ is defined as 
$\wt B:=w(\rank B)$.
\end{definition}

\begin{definition}[Chain Weight]
The \emph{weight} of a chain $C=(B_0,\dots,B_r)$
is defined as $\wt C := \sum_{h = 0}^{r} \wt B_h.$ 
We call $C$ {\em heavier} than another chain $C'$ if $\wt C>\wt C'$.
\end{definition}

\subsection{Block Proposals}

\definecolor{lightgray}{rgb}{0.9,0.9,0.9}
\definecolor{mygray}{rgb}{0.7,0.7,0.7}
\definecolor{darkgray}{rgb}{0.4,0.4,0.4}

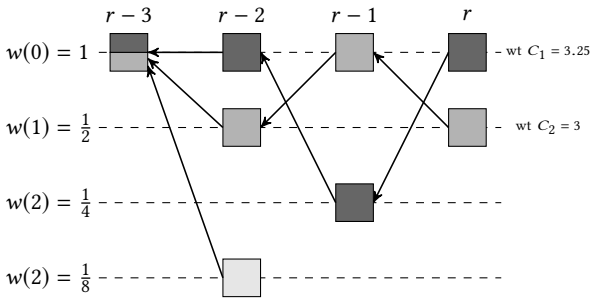
\begin{figure}
	\centering
\begin{tikzpicture}[scale=0.5,>=stealth',font=\sffamily] 
%rounds
\draw[color=black] (0,9) node {$r-3$};
\draw[color=black] (3,9) node {$r-2$};
\draw[color=black] (6,9) node {$r-1$};
\draw[color=black] (9,9) node {$r$};

%Ranks
\draw[color=black] (-3.5,8) node[right] {$w(0) = 1$};
\draw[color=black] (-3.5,6) node[right] {$w(1) = \frac{1}{2}$};
\draw[color=black] (-3.5,4) node[right] {$w(2) = \frac{1}{4}$};
\draw[color=black] (-3.5,2) node[right] {$w(2) = \frac{1}{8}$};

%Wt(C)
%\draw[color=black] (11.15,9.3) node[text width=4em, text centered] {Preferred Chain};
\draw[color=black] (11.15,8) node[text centered] {\tiny $\wt\ C_1 = 3.25$};

\draw[color=black] (11.15,6) node[text centered] {\tiny $\wt\ C_2 = 3$};
%grid
%\draw [color=cqcqcq,, xstep=8cm, ystep=1cm] (0,6) grid (8,4);
\draw[dashed] (-0.8,8) -- (9.9,8);
\draw[dashed] (-0.8,6) -- (9.9,6);
\draw[dashed] (-0.8,4) -- (9.9,4);
\draw[dashed] (-0.8,2) -- (9.9,2);

%Rectangles
%first column
\draw[fill=mygray] (-0.5,7.5) rectangle (0.5,8);
\draw[fill=darkgray] (-0.5,8) rectangle (0.5,8.5);

%second column
\draw[fill=darkgray] (2.5,7.5) rectangle (3.5,8.5);
\draw[fill=mygray] (2.5,5.5) rectangle (3.5,6.5);
\draw[fill=lightgray] (2.5,1.5) rectangle (3.5,2.5);

%Third column
\draw[fill=mygray] (5.5,7.5) rectangle (6.5,8.5);
\draw[fill=darkgray] (5.5,3.5) rectangle (6.5,4.5);

%Forth column
\draw[fill=darkgray] (8.5,8.5) rectangle (9.5,7.5);
\draw[fill=mygray] (8.5,6.5) rectangle (9.5,5.5);

%arrows
%2-->1
\draw [<-,line width=0.6pt] (0.5,8) -- (2.5,8);
\draw [<-,line width=0.6pt] (0.5,7.84) -- (2.5,6);
\draw [<-,line width=0.6pt] (0.5,7.65) -- (2.5,2);
%3-->2
\draw [<-,line width=0.6pt] (3.51,8) -- (5.5,4);
\draw [<-,line width=0.6pt] (3.5,6) -- (5.5,8);
%4-->3
\draw [<-,line width=0.6pt] (6.5,4) -- (8.5,8);
\draw [<-,line width=0.6pt] (6.5,8) -- (8.5,6);

\end{tikzpicture}
\caption{ Block proposals.
	The weight of each block is $1, 1/2, 1/4, 1/8, \ldots$,
    based on the rank of its proposer.
	Each chain accumulates the weights of its blocks
	(here shown only for blocks starting at round $r-3$).
	A replica with this view will resolve between the two forks that are still active at round $r$
	and will choose to build on the heavier one: $C_1$.}
\label{fig:fork}
\end{figure}

At each round each replica can propose a block.
To do so,
in round $r+1$,
the replica selects the heaviest valid\footnote{
Validity of blocks is explained after Def.~\ref{def:notarized}.
A chain is valid if and only if all its blocks are valid.}
chain $C$ with $\len C=r$ in its view
(cf.\ Fig.~\ref{fig:fork}).
The replica then considers all new transactions that it has received from users.
The new proposed block $B_r$ references $\head C$ and is composed of the selected transactions.
The replica broadcasts $B_r$ in order to request notarization from the notary committee.

\subsection{Block Notarization}

The goal of notarization is to enforce that chains are only being built from blocks
that were published during their respective round,
not later.
In other words,
notarization prevents that an adversary can build a private conflicting chain and reveal it later.
Blocks that are revealed too late cannot get notarized anymore,
so that the {\em timely publication} of block proposals is enforced.
A notarization is therefore regarded as a timestamp as well as a proof of publication.

The protocol guarantees to notarize at least one of multiple proposed chain heads for the current round.
It attempts to notarize exactly one chain head for each round but does not guarantee that.
Therefore, notarization does not imply consensus nor does it require consensus.

When participating in the notarization protocol
the replicas are only concerned about extending at least one valid chain,
not about which chain wins (which is the subject of the {\em finalization} in \S~\ref{sec:final}).

\begin{definition}\label{def:notarized}%[Notarization]
A {\em notarization} of block $B$ is an aggregated signature
by a majority subset of $U$ on the message $B$.
We call a block {\em signed} if it has received at least one signature
and {\em notarized} if it has received a notarization.
A {\em notarized block} is a block concatenated with a notarization of itself.
\end{definition}

As described in Alg.~\ref{alg:psp} below,
each replica in each round $r$ collects all valid block proposals
from all replicas (including from itself)
for a fixed time frame, the so-called $\btime$.
A proposed block $B$ is considered valid for round $r$
if $\round B=r$ and
there is a valid block $B'$ such that
\begin{enumerate}
	\item $\prev B = H(B')$ and $\round B'=\round B-1$,
	\item $\nota B$ is a notarization of $B'$,
	\item $\data B$ is valid.\footnote{
    Having a validity criteria for $\data B$ is optional and not required for the consensus protocol.
    Depending on the application of the blockchain, for example,
    $\data B$ can be configured to be valid only if $\data B$
    represents valid transactions and a valid state transition from $\data B'$.
    From the perspective of the consensus protocol $\data B$ is arbitrary data.}
\end{enumerate}
After $\btime$,
the replica signs all highest priority blocks for the current round
that it has received and broadcasts a signature message for this block to the entire network.

More than one block can have the highest priority
but only if the block maker has equivocated.
In this case all equivocated block proposals will get signed.
This is not an issue because each honest block makers in the next round
will only build on one of the blocks.

The replica continues to sign all highest priority block proposals
as it receives more block proposals after $\btime$.
When a notarization for the current round has been observed
then the replica advances to the next round.

\begin{algorithm}[h]
	\caption{-- Block Notarization}
	\label{alg:psp}
	\vspace{0.4em}
	\textbf{Goal: Notarize at least one block for the current round.} \\
	%\algFont 
    \vspace{2pt}
	\begin{algorithmic}[1]
%		\Statex \hspace{-\algorithmicindent} Upon receiving the start of the round $r$:
        \State Initialize chain with the genesis block
        \State $r\leftarrow 1$ \Comment {initialize round number}
        \While {true}
		\State $\mathsf{Wait}(\btime)$
		\While{no notarization for round $r$ received}
		\State $\calB\leftarrow$ set of all valid round-$r$ block proposals so far
        \For{All $B\in\calB$ with minimal $\rank B$}
        \If {$B$ not already signed}
	    \State $\sigma\leftarrow\sign(B)$
	    \State $\mathsf{Broadcast}(\sigma)$
	    \EndIf
        \EndFor
        \EndWhile
        \State $r\leftarrow r+1$
		\EndWhile
	\end{algorithmic}
	\label{alg:notary}	
\end{algorithm}

% ------------Example of how to use __________________
	%\If {} 
	%	\State  \label{ln:RecruitBridges}						
	%	\ForAll{$j \in [n]$}	\Comment{Distribute $d_i$ bridges}
	%		\State
	%		\State 			
	%	\EndFor			\label{ln:IterationEnd}
	%\Else
	%	\State 
	%	\State 
	%\EndIf \label{ln:IterationEnd}
	%\While{} \label{ln:ConditionSimple} %\Comment{}
	%	\State
	%\EndWhile

\subsection{Properties of notarization}
\subsubsection{Liveness}
From the description above it should be clear that Algorithm \ref{alg:psp} cannot deadlock
-- even in the presence of an adversary.
The fact that each replica continues to sign the highest priority block proposal
until a notarization for the current round is observed
is sufficient to ensure that at least one block gets notarized in the current round.
Eventually this will happen because \eqref{eq:2fU} holds and the ranking establishes a well-ordering on the set of block proposals.
After observing the first notarization for the current round
it is safe to stop signing because the observed notarization is being re-broadcasted and will eventually reach all honest replicas.
Thus, all honest replicas will advance to the next round.

The above argument relies on propagation assumptions for block proposals and notarizations.
We will analyze these in detail, including the relay policies involved, and provide a formal proof for liveness in \S~\ref{sec:security}.

\subsubsection{Honestly signed blocks}
\begin{definition}
A block is called {\em honestly signed} if it has received at least one signature from an honest replica.
\end{definition}
Note that an honest replica $i$ only issues a signature on a block $B$
if $B$ was the highest priority proposal visible to $i$ at some point in the round after $\btime$ expired.
The concept of honestly signed blocks is a theoretical one,
used to argue about the security properties of the notarization protocol.
It is not possible to tell if a given signature was issued by an honest or Byzantine replica.
Hence it is not observable whether a signed block is an honestly signed block or not.

\subsubsection{Timely publication}
\begin{definition}
An artifact of round $r$ was {\em timely published} (within $d$ rounds) 
if it was broadcasted while at least one honest replica was in round $\leq r$ (resp.\ in round $\leq r+d$).
\end{definition}
As a rule,
honest replicas re-broadcast every block that they sign.
Hence,
\begin{equation}\label{eq:published}
\text{\parbox{0.84\columnwidth}{
\flushleft
Only timely published blocks can be honestly signed.
}}
\end{equation}
Given a signed block, 
it is not possible to tell whether the block was timely published or not
because it is not possible to tell if the signature was honest or not.
This is different for notarized blocks as we will explain next.

\subsubsection{Notarization withholding}
By \eqref{eq:2fU},
any majority subset of replicas contains at least one honest replica.
Thus, we have:
\begin{equation}\label{eq:onlyhonest}
\text{\parbox{0.84\columnwidth}{
\flushleft
Only honestly signed blocks can get notarized.
}}
\end{equation}
We emphasize that an adversary can withhold its own signatures under an honestly signed block
which can lead to a situation where an honestly signed block does not appear to be notarized to the public.
However, the adversary can use its own signatures to produce and reveal a notarization
at any later point in time.

\subsubsection{Enforcing publication}
By \eqref{eq:onlyhonest} and \eqref{eq:published}, 
\begin{equation}\label{eq:onlypublished}
\text{\parbox{0.84\columnwidth}{
\flushleft
Only timely published block proposals can get notarized.
}}
\end{equation}
However, despite \eqref{eq:onlypublished}, notarizations can be withheld.
To show that withholding notarizations is harmless we introduce another theoretical concept:
\begin{definition}
A notarization $z$ is {\em referenced} if there is a notarized block $B$ with $\nota B=z$.
\end{definition}
Note that publishing a block, $B$, implicitly publishes the notarization of the previous block, $\nota B$,
because $\nota B$ is contained in $B$.
Hence, \eqref{eq:onlypublished} implies
\begin{equation}\label{eq:referenced}
\text{\parbox{0.84\columnwidth}{
\flushleft
Only notarizations that are timely published within $1$ round can get referenced.
}}
\end{equation}
Obviously, in a surviving chain all notarizations are referenced.
Thus, the publication of both block proposals and notarizations is enforced.
An adversary cannot build a private chain because a chain can only survive if:
\begin{itemize}
\item All its blocks were timely published.
\item All its notarizations were timely published within $1$ round.
\end{itemize}

\subsubsection{Consensus}\label{ssec:consensus}

We stated above that a chain can only survive
if all its notarizations were timely published within 1 round.
This means a replica looking at the notarizations of round $r$
can restrict itself to a certain time window.
All notarizations for round $r$ received after that window are necessarily irrelevant for the surviving chains.
This fact is the key to the finalization algorithm in \S~\ref{sec:final} below.
See Figure~\ref{fig:finalizing} for an example.

There are multiple ways that this fact can lead to a consensus point.
Note that it is not necessary for consensus that a single block gets notarized in a round,
nor that a single notarization can get referenced.
It is sufficient for a consensus point that all notarizations that were received within the time window 
(indirectly) reference the same block one (or more) rounds back.

\subsubsection{Normal operation}

\begin{definition}
A round has {\em normal operation} if only one block for that round gets notarized.
\end{definition}

Algorithm \ref{alg:psp} strives to achieve normal operation
by enforcing the $\btime$ waiting period,
and by giving preference to the highest priority block proposal.

If the highest priority block maker is honest and $\btime$ is large enough
then the highest priority block proposal will arrive at all honest replicas before $\btime$ expires.
This means that only one block can get notarized in this round.
Hence, assuming that $\btime$ is chosen correctly,
Algorithm \ref{alg:psp} will achieve normal operation in every round in which the highest priority block maker is honest.

We will analyze in detail what it means that $\btime$ is "large enough" (see \S~\ref{sec:security}),
taking into account the inner workings of the broadcast network such as the relay policy.
We will show that if the network traversal time is bounded by $\Delta$
then Alg.~\ref{alg:psp} is correct if $\btime\geq 3\Delta$
(Cor.~\ref{cor:progress}, Prop.~\ref{prop:growth}, Prop.~\ref{prop:quality}).

Note that every round with normal operation creates consensus on the unique block notarized in that round.
However, normal operation is a theoretical concept that cannot be observed 
due to the possibility of notarization withholding.
Luckily, as we saw in \S~\ref{ssec:consensus}, 
normal operation is not necessary for consensus.

\section{Finalization}\label{sec:final}
Replicas use the finalization procedure described in Alg.~\ref{alg:chain}
to identify points of consensus.
For this process,
it suffices to observe only the notarized blocks,
i.e.\ block proposals and individual signatures under block proposals can be ignored.
The finalization protocol is passive and independent from the notarization protocol.
Since it can be carried out by anyone (outside of the replicas) who has access to the notarized blocks,
we speak of {\em observers} in this section rather than of replicas.
\begin{algorithm}[ht]
	\caption{-- Finalization}
	\label{alg:chain}
	\vspace{0.4em}
	\textbf{Goal: Build the finalized chain from observing notarized blocks.} \\
	%\algFont 
    \vspace{2pt}
	\begin{algorithmic}[1]
		\Statex \hspace{-\algorithmicindent} {\sc Main}:
        \State $\calN_i\leftarrow\emptyset$ for all $i$ \Comment{Empty buckets for notarized blocks}
        \State $\calN_0\leftarrow\{B_0\}$ \Comment{Consider genesis block "notarized"}
        \State $C\leftarrow (B_0)$ \Comment{Finalized chain}
        \State $r\leftarrow 1$ \Comment{Current round}
        \While{true}
          \While{$\calN_{r}=\emptyset$}
            \State{$B\leftarrow$ incoming notarized block}
            \State{Store $B$ in $\calN_{\round B}$}
          \EndWhile
          \State{Schedule the call} \Call{Finalize}{$r-1$} at time $T$ from now
          \State{$r\leftarrow r+1$}
        \EndWhile
	\end{algorithmic}	
    \vspace{2pt}
	\begin{algorithmic}[1]
		\Statex \hspace{-\algorithmicindent} {\sc Finalize}$(h)$:
        \If{$h>0$}
        \State $C\leftarrow C(\calN_h)$
        \EndIf
	\end{algorithmic}	
\end{algorithm}

\subsection{Description}

Algorithm \ref{alg:chain} makes the assumption that
the observer receives
all round-$(r-1)$ notarizations that can get referenced
before time $T$ after having received the first notarization for round $r$.
This assumption is equivalent to the correctness of Alg.~\ref{alg:chain}
and is proved for all replicas in Theorem \ref{thm:alg2}.

The general idea of Alg.~\ref{alg:chain} is as follows:
We continuously collect all notarized blocks
and bucket them according to their round number.
Let $\calN_r$ be the bucket for all notarized blocks for round $r$.

Multiple buckets can be filled concurrently.
For example, a second block may go into $\calN_r$ even when $\calN_{r+1}$ is already non-empty.
However,
a block cannot be validated without knowing its predecessor.
Therefore, we assume that for every pair of blocks that reference each other,
the predecessor is processed first.
As a consequence $\calN_r$ must receive its first element before $\calN_{r+1}$ does.

By our initial assumption,
for each round $r$, there is a time
when we can rule out the receipt of any further notarized blocks for $\calN_r$ that can get referenced.
At that time we "finalize" round $r$
because we know that $\calN_r$ already contains all chain tips that can possibly survive
beyond round $r$.
Therefore, we can output the longest common prefix $C(\calN_r)$ as being final.

\definecolor{mygray}{rgb}{0.9,0.9,0.9}

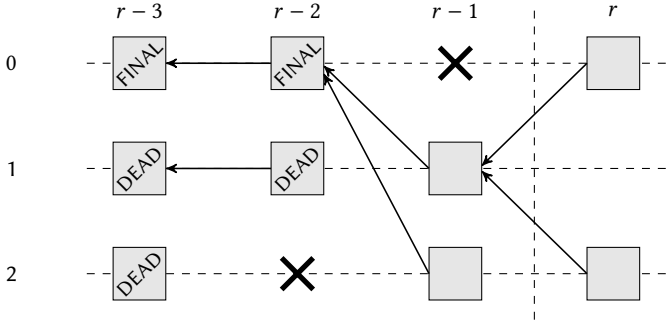
\begin{figure}
	\centering
\begin{tikzpicture}[scale=0.7,>=stealth',font=\sffamily] 
%rounds
\draw[color=black] (0,9) node {$r-3$};
\draw[color=black] (3,9) node {$r-2$};
\draw[color=black] (6,9) node {$r-1$};
\draw[color=black] (9,9) node {$r$};

%Ranks
\draw[color=black] (-2.7,8) node[right] {0};
\draw[color=black] (-2.7,6) node[right] {1};
\draw[color=black] (-2.7,4) node[right] {2};

%grid
%\draw [color=cqcqcq,, xstep=8cm, ystep=1cm] (0,6) grid (8,4);
\draw[dashed] (-1,8) -- (10,8);
\draw[dashed] (-1,6) -- (10,6);
\draw[dashed] (-1,4) -- (10,4);

%finalization cut-off
\draw[dashed] (7.5,9) -- (7.5,3);

%Rectangles
%first column
\draw[fill=mygray] (-0.5,7.5) rectangle node[rotate=45, font=\sffamily] {\small FINAL} (0.5,8.5);
\draw[fill=mygray] (-0.5,5.5) rectangle 
			node[rotate=45, font=\sffamily] {\small DEAD} (0.5,6.5);
\draw[fill=mygray] (-0.5,3.5) rectangle 
			node[rotate=45, font=\sffamily] {\small DEAD} (0.5,4.5);
%second column
\draw[fill=mygray] (2.5,7.5) rectangle node[rotate=45, font=\sffamily] {\small FINAL} (3.5,8.5);
\draw[fill=mygray] (2.5,5.5) rectangle 
			node[rotate=45, font=\sffamily] {\small DEAD} (3.5,6.5);
\draw[line width =2pt] 
			(2.70,4.3) -- (3.3,3.7)
			(2.7,3.7) -- (3.3,4.3);

%Third column
\draw [line width =2pt] 
			(5.7,7.7) -- (6.3,8.3)
			(6.3,7.7) -- (5.7,8.3);
\draw[fill=mygray] (5.5,5.5) rectangle (6.5,6.5);
\draw[fill=mygray] (5.5,4.5) rectangle (6.5,3.5);
%Forth column
\draw[fill=mygray] (8.5,8.5) rectangle (9.5,7.5);
\draw[fill=mygray] (8.5,4.5) rectangle (9.5,3.5);

%arrows
%2-->1
\draw [<-,line width=0.6pt] (0.5,8) -- (2.5,8);
\draw [<-,line width=0.6pt] (0.5,6) -- (2.5,6);
%3-->2
\draw [<-,line width=0.6pt] (3.5,7.95) -- (5.5,6);
\draw [<-,line width=0.6pt] (3.5,7.80) -- (5.5,4);
%4-->3
\draw [<-,line width=0.6pt] (6.5,6.05) -- (8.5,8);
\draw [<-,line width=0.6pt] (6.5,5.95) -- (8.5,4);

\end{tikzpicture}
\caption{Finalizing blocks.
	The diagram shows a view of notarized blocks at time
	$T$ after first seeing a notarized block for round $r$
	(i.e.\ at time $T$ after ending round $r$ and beginning round $r+1$).
    According to the correctness assumption, \eqref{eq:correct},
    no additional blocks that receive an inbound arrow
    can appear for the empty positions on the left of the vertical dashed line.
    Therefore, we mark positions $(r-2,2)$ and $(r-1,0)$ with a cross.
The $\textsc{Finalize}(r-1)$ procedure outputs the longest common prefix,
marked "FINAL",
of all chains defined by blocks from round $r-1$ as their chain tips.
As a consequence the blocks marked "DEAD" are now excluded from ever becoming final.
   }
\label{fig:finalizing}
\end{figure}

\subsection{Properties of Finalization}
We emphasize that there are some notable differences between the use of $\btime$ and $T$:
$\btime$ is agreed upon and part of the protocol specification,
whereas each observer can specify its own $T$.
The notarization protocol requires only $\btime$, not $T$.
The finalization protocol requires only $T$, not $\btime$.

The following assumption about Alg.~\ref{alg:chain} is called the {\em correctness assumption}:
\begin{equation}\label{eq:correct}
\text{\parbox{0.84\columnwidth}{
\flushleft
At the time when $\textsc{Finalize}(h)$ is being executed,
$\calN_h$ contains all %(system-wide) 
round-$h$ blocks that can get referenced.
}}
\end{equation}
\begin{proposition}\label{prop:correct}
Suppose \eqref{eq:correct} holds.
Then the chain $C$ in Alg.~\ref{alg:chain} is append-only.
\end{proposition}
The assertion justifies to call the chain $C$ in Alg.~\ref{alg:chain} {\em finalized}.
\begin{proof}
Let $C_h,C_{h+1}$ be the chains returned by $\textsc{Finalize}(h)$ and $\textsc{Finalize}(h+1)$, respectively,
in an execution of Alg.~\ref{alg:chain}.
Let $\calN^t_h$ denote the set $\calN_h$ at time $t$.
Obviously, $\calN^t_h$ can grow with increasing $t$.
Let $t_0, t_1$ be the times when $\textsc{Finalize}(h), \textsc{Finalize}(h+1)$ are called, respectively.
By \eqref{eq:correct},
none of the blocks added to $\calN_h$ after $t_0$ can get referenced.
In other words, we have
$\prev(\calN_{h+1}) \subseteq \calN^{t_0}_h$
regardless of the time at which $\calN_{h+1}$ is considered.
Thus,
\[
C_h=C(\calN^{t_0}_h)\leq C(\prev(\calN^{t_1}_{h+1})) \stackrel{\eqref{eq:prev}}{\leq} C(\calN^{t_1}_{h+1}) = C_{h+1}.
\]
\end{proof}

We will show that if the network traversal time is bounded by $\Delta$
then \eqref{eq:correct} holds 
if $T\geq 2\Delta$ (Thm.~\ref{thm:alg2}, Prop.~\ref{prop:consistency}).
Notably, this result does not make any assumptions about $\btime$,
i.e.\ the result holds even if the notaries choose an incorrect value for $\btime$.

There is alternative version of Alg.~\ref{alg:chain} which does not require the parameter $T$.
In line 10, instead of calling $\textsc{Finalize}(r-1)$ at time $T$ from now,
the alternative calls $\textsc{Finalize}(r-2)$ immediately.
This can also guarantee \eqref{eq:correct} according to Cor.~\ref{cor:final+3}
but it requires an assumption about the value of $\btime$ used by the notaries.

\section{Decentralized Randomness Beacon}\label{sec:drb}

The {\em decentralized random beacon} protocol (DRB)
allows replicas to agree on a verifiable random function (VRF)
and to jointly produce one new output of the VRF in every round.
By a VRF we mean a commitment to a deterministic, pseudo-random sequence $(\xi_r)_{r\geq 0}$
for which each output $\xi_r$ is unpredictable given the knowledge of all prior outputs $\xi_0,\ldots,\xi_{r-1}$
and for which each output $\xi_r$ is verifiable for correctness by anyone against the commitment.
In particular,
the VRF outputs are unbiasable due to their deterministic pseudo-random nature.
In our decentralized protocol,
the output $\xi_r$ shall not be predictable by the adversary
before at least one honest replica has advanced to round $r$.

Regarding the threat model,
we assume \eqref{eq:2fG} for all groups,
as stated in Assumption \ref{ass:2}.
For simplicity of exposition we describe the random beacon protocol
for a single group $G$ with $|G|=n$ and $n>2f(G)$.
The protocol can then be adapted to be executed by changing groups
as described in \S~\ref{sec:tr} below.

Our DRB protocol uses unique $t$-of-$n$ threshold signatures (see \S~\ref{sec:back})
created by the group $G$ as the source of randomness.
The adversary cannot predict the outcome of such a signature if $f\leq t-1$
and cannot prevent its creation if $f\leq n-t$.
If the adversary could abort the protocol by preventing a signature from being created,
then any restart or fallback mechanism would inevitably introduce bias into the output sequence.\footnote{
Several existing proposals in the literature are susceptible to bias due to a single party aborting the protocol.
For example Algorand \cite[\S~5.2]{giladalgorand} describes a fallback mechanism which inevitably introduces bias.
RANDAO \cite{randaorepo} relies on game-theoretic incentives
to keep malicious actors from aborting the protocol.
In practice, however, the gain for a malicious actor from biasing the randomness is unbounded whereas the penalty for aborting is bounded.
}
We treat the two failures (predicting and aborting) equally.
Therefore we require $t\in[f+1,n-f]$.
Note that if we set $n=2t-1$ then both conditions are equivalent to $f\leq t-1$.

The threshold signature scheme used in the DRB protocol is set up using a distributed key generation mechanism (see \ref{sec:dkg})
which does not rely on trusted parties.
We start by providing the background information on threshold cryptography that we use.

\subsection{Background on Threshold Cryptography}\label{sec:back}\hfill

\subsubsection{Threshold Signatures}
In a $(t,n)$-threshold signature scheme,
$n$ parties jointly set up a public key (the {\em group public key})
and each party retains an individual secret (the {\em secret key share}).
After this setup,
$t$ out of the $n$ parties are required and sufficient for creating a signature (the {\em group signature})
that validates against the group public key.

\subsubsection{Non-interactiveness}
A threshold signature scheme is called {\em non-interactive}
if the process of creating the group signature involves only a single round of one-way communication
for each of the $t$ participating parties.
Typically in a non-interactive scheme, 
each participating party creates a signature share using its individual secret
and sends this signature share to a third party.
Once the third party has received $t$ valid shares it can recover the group signature without any further interaction.
For example, ECDSA can be turned into a threshold signature scheme (\cite{gennaro2016threshold})
but it does not have the property of non-interactivity. 

\subsubsection{Uniqueness}
A signature scheme is called {\em unique} if for every message and every public key there is only one signature that validates successfully.
This property applies to single signature schemes and threshold signature schemes alike.
But in the setting of a threshold scheme it has the additional requirement that the signature must not depend on the subset of $t$ parties that participated in creating the signature.
In other words,
in a unique threshold signature scheme,
regardless of who signs, the resulting group signature will always be the same.

"Unique" is a property that is strictly stronger than "deterministic".
A signature scheme is called {\em deterministic} if the signing function does not use randomness.
Note that "unique" is a property of the verification function
whereas "deterministic" is a property of the signing function.
Unique implies deterministic but not conversely.
For example, DSA and ECDSA can be made deterministic by re-defining the signing function
in a way that it derives its so-called "random $k$-value" deterministically via a cryptographic hash function 
from the message plus the secret key
instead of choosing it randomly.
However, this technique cannot be used to make DSA or ECDSA unique
because one cannot expose the $k$-value to the verification function.

\subsubsection{Distributed Key Generation (DKG)}\label{sec:dkg}
For a given $(t,n)$-thres\-hold signature scheme,
a DKG protocol allows a set of $n$ parties to collectively generate
the keys required for the scheme
(i.e.\ the group public key and the individual secret key shares)
without the help of a trusted party.

Note that a DKG protocol is more than a secret sharing protocol.
In a secret sharing protocol the secret shares can be used to recover the group secret,
but this can be done only once.
After everyone has learned the group secret the shares are not reusable.
In a DKG the shares can be used repeatedly for an unlimited number of group signatures
without ever recovering the group secret key explicitly.

DKG protocols are relatively straight-forward for discrete-log based cryptosystems
and typically utilize multiple instances of a verifiable secret sharing protocol (VSS).
\dfn\ uses the "Joint-Feldman DKG"%
\footnote{
It is known from \cite{gennaro2016threshold} that the adversary can bias the distribution of public keys generated by the Joint-Feldman DKG.
However, the bias generally does not weaken the hardness of the DLP for the produced public key (\cite[\S~5]{gennaro2016threshold}).
Therefore, with the simplicity of our protocol in mind,
we use the original, unmodified Joint-Feldman DKG
even though variations are available that avoid the bias.
}
as described in \cite{gennaro1999}.

\subsection{The BLS signature scheme}
The only known signature schemes that have a unique, non-interactive threshold version
and allow for a practical, efficient DKG are the pairing-based schemes
derived from BLS~\cite{BLS}.
BLS was introduced by Boneh, Lynn, and Shacham in 2003
and related work can be found in \cite{libert2016born}.
We shall use the original BLS scheme throughout.

\subsubsection{BLS functions}
Assuming we have generated a secret/public key pair $(\sk,\pk)$,
BLS provides the following functions:
 \begin{enumerate} 
 	\item $\sign(m, \sk)$: Signs message $m$ using secret key $\sk$ and returns signature $\sigma$.
 	\item $\verify(m, \pk, \sigma)$: Verifies the signature $\sigma$ for message $m$ against the public key $\pk$ and returns true or false.
 \end{enumerate}
Under the hood,
BLS uses a non-degenerate, bilinear pairing
$$e: \GG_1 \times \GG_2 \rightarrow \GG_T$$
between cyclic subgroups $\GG_1,\GG_2$ of suitable elliptic curves points
with values in a group of units $\GG_T$.
We shall write all groups multiplicatively in this paper.
For each group,
we fix an arbitrary generator: $g_1 \in \GG_1$, $g_2 \in \GG_2$, $g_T \in \GG_T$.
We also assume a hash function
$\hash_1: \{0,1\}^{\ast} \rightarrow \GG_1$
with values in $\GG_1$.

The secret keys are scalars, the public keys are elements of $\GG_2$
and the signatures are elements of $\GG_1$.
The function $\sign(m,\sk)$ computes $\hash_1(m)^\sk$ 
and $\verify(m, \pk, \sigma)$ tests whether $e(\sigma,g_2)=e(\hash_1(m),\pk)$.

\subsubsection{Threshold BLS}
We refer to the threshold version of BLS as {\em TBLS}.
The same functions $\sign$ and $\verify$ that are defined for BLS
also apply to the key/signature shares and group keys/signatures in TBLS.
We assume all participating parties in the $(t,n)$-DKG are numbered $1,\ldots,n$.
After having run the DKG as in \ref{sec:dkg},
the $(t,n)$-TBLS provides additionally the function:
 \begin{enumerate}
 	\item $\recover(i_1,\ldots,i_t,\sigma_{i_1}, \dots, \sigma_{i_t})$: 
    Recover the group signature $\sigma$ from the signature shares $\sigma_{i_j}$, $j=1,\ldots,t$,
    where $\sigma_{i_j}$ is provided by the party $i_j\in\{1,\ldots,n\}$.
 \end{enumerate}
Because of the uniqueness property,
the output of $\recover$ does not depend on which $t$ shares from the group are used as inputs.
$\recover$ computes a "Lagrange interpolation" for points in $\GG_1$.
The indices $i_1,\ldots,i_t$ must be pairwise different for the $\recover$ function to succeed.

\subsection{Randomness Generation}\label{sec:drb-setup}

The randomness generation consists of 
a) a one-time setup in which a DKG is run
and 
b) a repeated signing process in which the outputs are produced.
The DKG is slow and requires agreement
whereas the repeated signing is non-interactive and fast.

\subsubsection{Setup}
\label{sec:setup}

When setting up a threshold signature scheme,
we do not want to rely on any trusted third party.
Therefore, the group $G$ runs a DKG for BLS
to set up the group public key and the secret key shares during the initialization of the blockchain system.
The threshold $t$ is a parameter of the setup.

Once the DKG has finished successfully, it
outputs a public verification vector $V_\grp\in\GG_2^t$, 
and leaves each replica $i\in \grp$ with its secret key share $\sk_{G,i}$. 
The verification vector $V_\grp$ gets committed to and recorded in the blockchain,
for example in the genesis block.

Let $V_G=(v_0,\ldots,v_{t-1})$.
The group public key is $\pk_\grp = v_0\in\GG_2$.
The secret key $\sk_\grp$ corresponding to $\pk_\grp$ is not known
explicitly to anyone in $\grp$ but can be used implicitly through
$\sk_{\grp,i}$. The verification vector $V_\grp$ can be used to recover the public key share
$\pk_{\grp,i}\in\GG_2$ corresponding to $\sk_{G,i}$ via ``polynomial'' substitution 
$$\pk_{G,i} = \prod_{k=0}^{t-1}v_k^{i^k}\in\GG_2.$$
Hence, all signature shares produced by $i$ 
can be publicly verified against the information $V_\grp$ and $i$.
The group public key $\pk_\grp$ can be used to verify the output of $\recover$.

\subsubsection{Signing process}\label{sec:beaconsign}
\label{sec:drb-dvrf}
Recall that a replica enters round $r$ upon seeing the first notarization for round $r-1$.
At the beginning of its round $r$,
replica $i \in G$ computes the signature share
$$\sigma_{r,i} = \sign(r \parallel \rnd_{r-1},\sk_{G,i}), $$
where $\rnd_{r-1}$ is the random value of round $r-1$.
To bootstrap, $\rnd_{0}$ has been set to a nothing-up-my-sleeve number,
\eg\ the hash of the string ``\dfn''.
Replica $i$ then broadcasts $(i,\sigma_{r,i})$.

Any replica who receives this data can validate $(i,\sigma_{r,i})$
against the public information $V_G$ as described in \ref{sec:setup} above.
If valid then the replica stores and re-broadcasts $(i,\sigma_{r,i})$.
As soon as a replica has received at least $t$ different valid signature shares,
it runs $\recover(i_1,\ldots,i_t,\sigma_{r,i_1}, \dots, \sigma_{r,i_t})$
to compute the group signature $\sigma_{G,r}$.
Finally, the random output $\rnd_r$ for round
$r$ is computed as the hash of $\sigma_{G,r}$.

We emphasize that the signing process is non-interactive.
Any third party can do the recovery after a one-way communication of sufficiently many shares.

\section{Scalability}\label{sec:scale}
\subsection{Threshold Relay}\label{sec:tr}
For reasons of scalability
the notarization and random beacon protocols from \S~\ref{sec:psp} and \S~\ref{sec:drb}
are executed by groups of size $n$ rather than by all replicas in $U$.
Otherwise the message complexity would be unbounded as the total number of replicas grows.

The groups, also called {\em committees} here,
are random samples of size $n$ from the whole population $U$.
The group size $n$ is a system parameter that is chosen according to the failure probability analysis of \S~\ref{ssec:grpsize}.
A large enough group size ensures that 
-- up to an acceptable failure probability --
every group used in the system is honest (Assumption \ref{ass:2}).

The mechanism by which \dfn\ randomly samples replicas into groups,
sets the groups up for threshold operation,
chooses the current committee,
and relays from one committee to the next
is called {\em threshold relay}.

\subsubsection{Group Derivation}\label{sec:derive}

Let $n$ be the group size.
The groups are derived from a random seed $\xi$
where the $j$-th derived group is
\begin{equation}
\group(\xi,j) := \perm_U(\prng(\xi,j))(\{1,\ldots,n\}).
\end{equation}

At the start of the system, 
we choose a number $m$ and a seed $\xi$ and form groups
\begin{equation}
G_j := \group(\xi,j), \quad j=1,\ldots,m.
\end{equation}
Each $G_j$ runs the DKG described in Section ~\ref{sec:drb-setup} 
to create group keys $\pk_{G_j}$ which are then stored in the genesis block.

\subsubsection{Committee Selection}\label{sec:select}
The sequence $(\xi_i)$ is bootstrapped by defining an initial value for $\xi_0$.
Then, in round $r$, we choose
\begin{equation}
G^{(r)}:=G_j, \quad j:=\xi_r \bmod m
\end{equation}
as the committee for round $r$.
The same committee can be used for the notarization and the random beacon protocols of the same round.

In the random beacon protocol,
the members of $G^{(r)}$ jointly produce the output $\xi_r$,
which is then used to select the next committee $G^{(r+1)}$.
Since activity is relayed from one group to the next,
we call the mechanism "threshold relay".

\subsection{Open Participation}\label{sec:open}

\usetikzlibrary{arrows}

\definecolor{mygray}{rgb}{0.7,0.7,0.7}
\definecolor{darkgray}{rgb}{0.2,0.2,0.2}

\tikzset{Block/.style={draw, minimum size=8,line width=1, thick, draw=black!80, fill = mygray}}
\tikzset{KeyBlock/.style={draw, minimum size=8,line width=1, thick, draw=black!80, fill = darkgray}}
\tikzset{Arrow/.style={-stealth, thick, draw=black!80}}

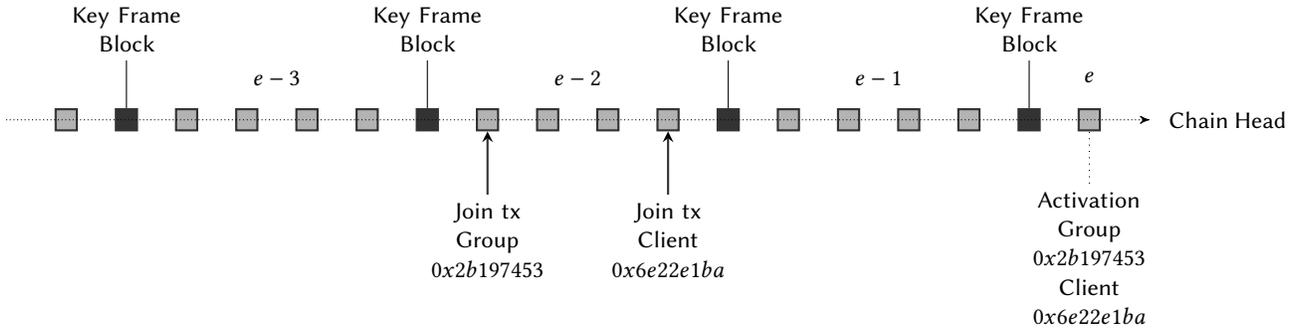
\begin{figure*}
	
\begin{tikzpicture} [scale=0.8,>=stealth',font=\sffamily] 

%Blocks
\node [Block] (v1) at (1,4) {};
\node [KeyBlock] (v2) at (2,4) {};
\node [Block] (v3) at (3,4) {};
\node [Block] (v4) at (4,4) {};
\node [Block] (v5) at (5,4) {};
\node [Block] (v6) at (6,4) {};
\node [KeyBlock] (v7) at (7,4) {};
\node [Block] (v8) at (8,4) {};
\node [Block] (v9) at (9,4) {};
\node [Block] (v10) at (10,4) {};
\node [Block] (v11) at (11,4) {};
\node [KeyBlock] (v12) at (12,4) {};
\node [Block] (v13) at (13,4) {};
\node [Block] (v14) at (14,4) {};
\node [Block] (v15) at (15,4) {};
\node [Block] (v16) at (16,4) {};
\node [KeyBlock] (v17) at (17,4) {};
\node [Block] (v18) at (18,4) {};

%texts
\node[text width=1.5cm, text centered] (t1) at (2,5.5) {Key Frame Block};
\node[text width=1.5cm, text centered] (t2) at (7,5.5) {Key Frame Block};
\node[text width=1.5cm, text centered]  (t3) at (12,5.5) {Key Frame Block};
\node[text width=1.5cm, text centered]  (t4) at (17,5.5) {Key Frame Block};

\node (t5) at (4.5,4.7) {$e-3$};
\node (t6) at (9.5,4.7) {$e-2$};
\node (t7) at (14.5,4.7) {$e-1$};
\node (t8) at (18,4.7) {$e$};

\node (t9) at (20.3,4) {Chain Head};

\node[text width=1.5cm, text centered]  (t10) at (8,2) {Join tx Group $0x2b197453$};
\node[text width=1.5cm, text centered]  (t11) at (11,2) {Join tx Client $0x6e22e1ba$};
\node[text width=1.5cm, text centered]  (t12) at (18,1.7) {Activation Group $0x2b197453$ \newline Client $0x6e22e1ba$};

%arrows
\draw (t1) -- (v2);
\draw (t2) -- (v7);
\draw (t3) -- (v12);
\draw (t4) -- (v17);

\draw [Arrow] (t10) edge (v8);
\draw [Arrow] (t11) edge (v11);
\draw [dotted] (t12) -- (v18);

%grid
\draw[densely dotted, ->] (0,4) -- (19,4);

\end{tikzpicture}
\captionsetup{font=footnotesize}
\caption{Epochs and Registration.
The chain is divided into epochs defined by the round numbers of the blocks.
A client joins by submitting a special transaction into a block
which also locks up a stake deposit.
A group joins by successfully executing a DKG (distributed key generation)
and submitting the result as a join transaction into the blockchain.
Clients and groups become actively involved in the protocol only after a gap of at least 1 epoch
between their join transaction and their first activity.
}
\label{fig:epochs}
\end{figure*}

It is impractical to assume that the set of all replicas is known from the start of the protocol,
especially in \dfn's public chain. 
This section describes how the protocol adopts an open participation model
in which new replicas can join and existing replicas can leave the system.

\subsubsection{Epochs}
We divide the rounds into non-overlapping {\em epochs} of length $l$
where $l$ is a system parameter and is fixed. 
The block produced in the first round of each epoch is a {\em registry block} (also called {\em key frame})
and contains a summary of all new registrations and de-registrations of replicas
that happened during the previous epoch that just ended.
Note that the summary is a deterministic result of all the blocks in the preceding epoch
so that the block maker of the key frame has no opportunity to censor registrations.
The first round of the very first epoch is DFINITY's genesis block which is also a key frame.

\subsubsection{Registration of Replicas}
A replica can request to join the network (i.e.\ {\em register}) or leave the network (i.e.\ {\em de-register})
by submitting a special transaction for that purpose.
The transaction is submitted to the existing replicas for inclusion in the chain
just like any other user transaction.
A registration transaction contains the public key of the new replica
plus an {\em endorsement} proving that it was allowed to form.
Depending on the underlying Sybil resistance method,
the endorsement is, e.g., the proof of a locked-up stake deposit,
the solution to a proof-of-work puzzle tower,
or the certification by a central, trusted authority.

\subsubsection{Registration of Groups}
The random beacon output of the first round in an epoch $e$
defines the composition of all the groups
that are allowed to newly enter the system during this epoch.
A system parameter, $m_{\max}$, governs how many different groups can form during an epoch.

Let $r$ be the first round of epoch $e$.
For each $j\leq m_{\max}$ the $j$-th {\em candidate group}
is defined as $G=\group(\xi_r,j)$.
The members of $G$ run a DKG to establish a group public key $\pk_G$.
If the DKG succeeds then the members create a registration transaction for $G$
which contains the tuple $x=(e,j,\pk_G)$.
After $x$ is signed by a super-majority of $G$,
any member can submit $x$ for inclusion in the blockchain.
The validity of the signature under $x$ is publicly verifiable
against the information already on the blockchain,
i.e., the pool $U$ of active replicas and the random beacon output $\xi_r$ that defined the group.
A registration transaction $x$ is only valid if it is included in a block 
that lies within the epoch $e$.

If the DKG fails
or $x$ fails to get a super-majority signature from $G$
or $x$ is not included in the blockchain within epoch $e$
then $G$ cannot register.
An adversary can cause the registration to fail.
For example, if the super-majority is defined as a $\nicefrac{2}{3}$-majority
then an adversary controlling $\geq\nicefrac{1}{3}$ of $G$ can deny the signature under $x$.
However, due to variance, this will happen only to some of the group candidates.
For example, an adversary controlling $<\nicefrac{1}{3}$ of $U$
will control $<\nicefrac{1}{3}$ in at least half of all groups.

Groups are de-registered automatically
when they expire after a fixed number of epochs defined by a system parameter.

\subsubsection{Delayed Activity}
If the registration of a new identity (replica or group) is included in the chain in epoch $e$,
then the newly registered entity becomes active in epoch $e+2$.
Thus, there is always a gap of at least $l$ rounds
between the registration of a new entity and the first activity of that new entity.
This sequence of events is shown in Fig.~\ref{fig:epochs}.

The gap is required to ensure that all registrations of new entities are finalized
before they can be allowed to have any influence on the random beacon.
The minimum value of $l$ can be derived from the growth property of the finalized chain
that is proved in Prop.~\ref{prop:growth} below.

\dfn\ uses a value for $l$ that is far greater than the minimum required, 
because we want to limit the rate at which key frames are produced, 
in order to reduce load on so-called observing "light clients".

\section{Security Analysis}\label{sec:security}
In this section,
we show that the \dfn\ protocol provides us with a robust and fast distributed public ledger abstraction.
Any ledger must satisfy the following two fundamental properties
which we will derive from lower-level properties in \S~\ref{sec:chainprop}.
\begin{definition}[Ledger properties]\mbox{}
\begin{enumerate}[a)]
\item {\em Persistence.}
Once a transaction is included into the finalized chain of one honest replica,
it will then be included in every honest replica's finalized chain.
\item {\em Liveness.} All transactions originating from honest parties will eventually be included in the finalized chain.
\end{enumerate}
\end{definition}
What distinguishes \dfn\ from other ledgers is the property of {\em near-instant finality}.
This property is formalized by the following two definitions and theorem.

\begin{definition}[Number of Confirmations]
We say a transaction has $n$ {\em confirmations} 
if it is contained in a notarized block $B_r$
and there is a chain of notarized blocks of the form $(\ldots, B_r,\ldots, B_{r+n-1})$.
\end{definition}
Note that the definition refers to any notarized blocks known to the replica,
not necessarily finalized blocks.

\begin{theorem}[Main Theorem]\label{thm:main}
Under normal operation in round $r$,
every transaction included in a block for round $r$ is final
after two confirmations plus the maximum network roundtrip time $2\Delta$.
\end{theorem}
From the perspective of an arbitrary observer,
the Main Theorem means the following.
Suppose an observer sees a transaction $x$ that has received two confirmations,
i.e.\ a notarized block $B_r$ for round $r$ containing $x$ 
and another notarized block $B_{r+1}$ with $\prev B_{r+1}=B_r$.
If round $r$ experienced normal operation
then,
at time $2\Delta$ after the observer received the notarization for $B_{r+1}$,
the finalization algorithm (Alg.~\ref{alg:chain}) is guaranteed to
append $B_r$ to the observer's final chain.
We assume here that $T$ in Alg.~\ref{alg:chain} is set to $2\Delta$.
The proof of the Main Theorem will occupy \S~\ref{sec:instant} below.

We will provide proofs for the synchronous model where an upper bound for the network traversal time $\Delta$ is known.
We assume that processing times for messages are included in the network traversal time.

\subsection{Broadcast and Processing}\label{ssec:relay}

The security analysis must take into account the behavior of the broadcast network,
which implements a gossip protocol.
In particular, the relay policy that is applied to gossiping
is going to be essential for the provability of our results.

Replicas continuously receive new protocol artifacts,
e.g.\ block proposals, signatures under block proposals, notarizations,
notarized blocks or random beacon outputs.
As soon as an artifact is determined to be valid
it is immediately relayed ("gossiped") to the replica's peers
if it falls under the relay policy defined below.
\begin{definition}[Relay Policy]\label{def:relay}
All honest replicas relay the following artifacts
\begin{enumerate}[a)]
\item for the current round:
valid block proposals and valid signatures under block proposals,
\item for any round:
notarizations and notarized blocks.
\end{enumerate}
We say an artifact has {\em saturated the network}
if it has been received by all honest replicas.
\end{definition}
We emphasize that saturating the network is a global condition that is only of theoretical value in our security arguments.
The replicas cannot observe whether an artifact has saturated the network or not.
Saturating the network does not constitute a reliable broadcast.

Artifacts can be received out of order,
e.g.\ a signature or a notarization for a block can be received before the block.
If an artifact $x$ is received before an artifact that is referenced by $x$ then $x$ can not be validated.
For this reason,
all honest replicas first queue any incoming artifact $x$
until all artifacts referenced by $x$ have also been received.
Only then is $x$ processed.
In particular, an honest replica $i$ relays an artifact $x$
only if it possesses all artifacts referenced by $x$.
Hence, a peer $j$ of $i$ who receives $x$ from $i$ can then request any artifact 
that is referenced by $x$ that $j$ does not already possess.
This is an artifact synchronization process which happens transparently in the background
and is completed before $j$ processes $x$.
Therefore, throughout the paper,
we take for granted that if an artifact $x$ is received
then all artifacts referenced by $x$ have also been received.

Signatures under block proposals are collected in a background process and,
once a majority is available for a given block proposal,
are aggregated into a notarization
which is then treated in the same way as if it was received from outside.
Block proposals and notarizations are collected in the background and made available to Alg.~\ref{alg:psp} and \ref{alg:chain}.

Our relay policy and network assumptions (see \S~\ref{ssec:prel} below)
guarantee the following property:
\begin{equation}\label{eq:relayb}
\text{\parbox{0.84\columnwidth}{
\flushleft
Any artifact that falls under b) and is processed by an honest replica
will eventually saturate the network.
}}
\end{equation}
Property \eqref{eq:relayb} does not hold for artifacts relayed under policy a)
because a replica in the middle of the broadcast path may have advanced to the next round
in which case the artifact will be considered old and will be dropped.

Suppose an honest replica $i$ has processed a block proposal $B$ and considered it valid.
Then $i$ must possess $\prev B$ and a notarization of $\prev B$.
We emphasize that the honest replica $i$ re-broadcasts the notarized block $\prev B$ in this case.
By \eqref{eq:relayb}, this behavior guarantees:
\begin{equation}\label{eq:B}
\text{\parbox{0.84\columnwidth}{
\flushleft
If a block $B$ is honestly signed
then the notarized block $\prev B$ eventually saturates the network.
}}
\end{equation}
Property \eqref{eq:B} does not hold for $B$ itself
or for individual signatures under $B$,
for these artifacts do not fall into category b) of Def.~\ref{def:relay}.

\subsection{Timing and Progress}
This section makes statements about the relative timing of events that happen at different replicas.
We do not assume normal operation in any round
and therefore have to consider the possibility of multiple notarizations $z_r,z'_r,\ldots$
being created and broadcasted for the same round $r$.

\subsubsection{Preliminaries}\label{ssec:prel}
We assume that a message broadcasted by an honest replica at time $t$
reaches every honest replica before $t+\Delta$ (i.e.\ at a time $<t+\Delta$).
Since processing times are not in the scope of our analysis,
we assume all processing times to be zero.
This applies to the creation as well as to the validation of all messages
including block proposals, signatures, notarizations, random beacon shares, and random beacon outputs.
As a consequence, for example,
when a replica $i$ receives a random beacon output $\xi_r$ at time $t$
then it broadcasts its block proposal for round $r$ at the same time $t$.
Or, when a random beacon member $i$ receives a notarization $z_r$ for round $r$ at time $t$
then $i$ broadcasts its random beacon share for round $r+1$ immediately at the same time $t$.

\begin{definition}
Let $\tau_i(A)$ denote the time at which replica $i$ sees event $A$ where $A$ is one of the following:
a random beacon output $\xi_r$,
a block proposal $B_r$, 
or a notarization $z_r$.
We set $\tau_i(r):=\tau_i(z_{r-1})$ where $z_{r-1}$ is the first notarization for round $r-1$ that $i$ receives.
\end{definition}
Thus, $\tau_i(r)$ is the time when replica $i$ enters round $r$.
To study when the first honest or last honest replica sees an event, we define:
\begin{definition}
\mbox{}\[
\ubar{\tau}(A) := \min_{i\textrm{ honest}}\tau_i(A), \quad
\bar{\tau}(A):= \max_{i\textrm{ honest}}\tau_i(A).
\]
\end{definition}
For example,
$\ubar{\tau}(r)$ is the time when the first honest replica enters round $r$ and
$\bar{\tau}(r)$ is the time when the last honest replica enters round $r$.
Finally, it is also of interest when an event can first be seen or constructed by the adversary.
Therefore, we define:
\begin{definition}
\mbox{}\[
\utau^*(A) := \min_{i}\tau_i(A).
\]
\end{definition}
For example, 
$\utau^*(\xi_r)$ is the earliest time that the adversary can construct the random beacon output $\xi_r$.

We will prove in Cor.~\ref{cor:progress} below that the protocol makes continuous progress,
i.e.\ that all values $\tau_i(r)$ are finite.
As the reader may verify,
the statements made in this section up until Cor.~\ref{cor:progress}
also hold (trivially) in the case that any of the values $\tau_i(A)$ are infinite.

\begin{lemma}
For all rounds $r$ we have:
\begin{align}\label{eq:tauij}
\otau(r)\leq\utau(r)+\Delta
\end{align}
and,
for any round-$r$ event $A$ under Def.~\ref{def:relay}{a)},
\begin{align}\label{eq:tauijA}
\utau(A)+\Delta\leq\utau(r+1) \Longrightarrow \otau(A)\leq\utau(A)+\Delta
\end{align}
\end{lemma}
\begin{proof}
Let $i$ be an honest replica and let $z_{r-1}$ be a notarization for round $r-1$ 
such that $\tau_i(z_{r-1})=\utau(r)$.
By Def.~\ref{def:relay}{b)}, $z_{r-1}$ is relayed across the network 
and reaches any other honest replica $j$ by $\tau_i(r)+\Delta$.
This proves \eqref{eq:tauij}.

Now let $A$ be any event that falls under the relay policy in Def.~\ref{def:relay}{a)} for round $r$.
If $\utau(A)+\Delta\leq\utau(r+1)$
then the same argument applies.
Indeed, the assumption means that all replicas along the broadcast path will still be in round $r$,
thus will relay $A$ according to Def.~\ref{def:relay}{a)}.
This proves \eqref{eq:tauijA}.
\end{proof}

\subsubsection{Maximal Progress}
\begin{lemma}[Notarization, "Fast" Bound]
For all rounds $r$ we have:
\begin{align}\label{eq:taubtime}
\utau(r)+\btime\leq\utau^*(r+1).
\end{align}
\end{lemma}
\begin{proof}
Honest replicas participate in a notarization $z_r$ for round $r$ only after they have been in round $r$ for at least $\btime$ (cf.~Alg.~\ref{alg:psp}).
The inequality \eqref{eq:taubtime} means that at least one honest replica is required
before anyone can see a notarization $z_r$ for round $r$.
\end{proof}
\begin{proposition}[Maximal Progress]\label{prop:minround}
Suppose $\btime\geq \Delta$.
Then the round number of any honest replica increases at most every $\btime - \Delta$.
Moreover, at any point in time,
the difference between the round numbers of two honest replicas can be at most $1$.
\end{proposition}
\begin{proof}
From \eqref{eq:tauij} and \eqref{eq:taubtime} together we get:
$$\otau(r)+(\btime-\Delta)\leq\utau^*(r+1).$$
This implies both statements.
\end{proof}

\begin{corollary}[Safe Broadcast]\label{cor:fwd}
Suppose $\btime\geq\Delta$.
For all rounds $r$ and any round-$r$ event $A$ under Def.~\ref{def:relay}{a)} we have:
\begin{align}\label{eq:tauA}
\utau(A)\leq\utau(r)+(\btime-\Delta) \Longrightarrow \otau(A)\leq\utau(r)+\btime.
\end{align}
\end{corollary}
The interpretation is that if a round-$r$ event is broadcasted by $\utau(r)+(\btime-\Delta)$
then it is guaranteed to saturate the network,
and this happens by $\utau(r)+\btime$.
\begin{proof}
From $\utau(A)\leq\utau(r)+(\btime-\Delta)$ we conclude
$$\utau(A)+\Delta\leq\utau(r)+\btime\stackrel{\eqref{eq:taubtime}}{\leq}\utau(r+1).$$
Thus, by \eqref{eq:tauijA},
$$\otau(A)\leq\utau(A)+\Delta\leq\utau(r)+\btime.$$
\end{proof}
\subsubsection{Normal Operation}
\begin{lemma}[Beacon, "Slow" Bound]
For all rounds $r$ we have:
\begin{align}\label{eq:tauxi}
\otau(\xi_{r})\leq\utau(r)+2\Delta
\end{align}
\end{lemma}
\begin{proof}
Each honest random beacon member $i$
broadcasts its random beacon share for $\xi_{r}$ at $\tau_i(r)$ (cf.~\S~\ref{sec:beaconsign}).
Thus, any other honest replica will receive all honest random beacon shares for $\xi_{r}$ by
$\otau(r)+\Delta$
and thus will recover $\xi_r$ by that time, i.e.\
$$
\otau(\xi_{r})\leq\otau(r)+\Delta.
$$
The assertion follows after applying \eqref{eq:tauij}.
\end{proof}

We now assume $\btime\geq 3\Delta$ for the rest of the subsection.
The timing of events with $\btime=3\Delta$ is illustrated in Fig.~\ref{fig:events}.
\definecolor{rvwvcq}{rgb}{0.08235294117647059,0.396078431372549,0.7529411764705882}
\definecolor{cqcqcq}{rgb}{0.7529411764705882,0.7529411764705882,0.7529411764705882}

\begin{figure}
	\centering
\begin{tikzpicture}[scale=0.8] 
%replicas
\draw[color=black] (-2,6) node[right] {Replica $i$};
\draw[color=black] (-2,5) node[right] {Random Beacon};
\draw[color=black] (-2,4) node[right] {Block Makers};
\draw[color=black] (-2,3) node[right] {Replica $j$};

%top timings
\draw[decorate, decoration={brace,amplitude=5pt}] (2,6) -- node[above=0.7ex] {\footnotesize $\btime = 3\Delta$} (8,6);

\draw [->,line width=0.6pt] (2,6.5) -- node[above = 0.7ex, text width=2.1cm] {\footnotesize $i$ starts the waiting period of Alg.~\ref{alg:psp}} (2,6.05);

\draw [->,line width=0.6pt] (8,6.5) -- node[above = 0.7ex, text width=1.5cm] {\footnotesize $i$ is ready to sign} (8,6.05);

%botom timings 
\draw [->,line width=0.6pt] (2,3) -- node[below = 0.7ex] {\footnotesize $\utau(r)=\tau_i(r)$} (2,2.75);

\draw[decorate, decoration={brace,amplitude=5pt,mirror}] (2,3) -- node[below=0.7ex] {\footnotesize $\Delta$} (4,3);

\draw[decorate, decoration={brace,amplitude=5pt,mirror}] (4,3) -- node[below=0.7ex] {\footnotesize $\Delta$} (6,3);

\draw[decorate, decoration={brace,amplitude=5pt,mirror}] (6,3) -- node[below=0.7ex] {\footnotesize $\Delta$} (8,3);

%arrows
\draw [->,line width=0.6pt] (2,6) -- node[above=-0.7ex, rotate= -25] {\footnotesize $z_{r-1}$} (3.9,5.05);
\draw [->,line width=0.6pt] (2,6) -- node[above=-0.7ex, rotate= -40] {\footnotesize $z_{r-1}$} (3.9,4.05);
\draw [->,line width=0.6pt] (2,6) -- node[above=-0.7ex, rotate= -50] {\footnotesize $z_{r-1}$} (3.9,3.05);

\draw [->,line width=0.6pt] (4,5) -- node[above=-0.7ex, rotate= 25] {\footnotesize $\rnd_r$} (6,5.95);
\draw [->,line width=0.6pt] (4,5) -- node[above=-0.7ex, rotate= -25] {\footnotesize $\rnd_r$} (5.9,4.05);
\draw [->,line width=0.6pt] (4,5) -- node[above=-0.7ex, rotate= -40] {\footnotesize $\rnd_r$} (5.9,3.05);

\draw [->,line width=0.6pt] (6,4) -- node[above=-0.7ex, rotate= 50] {\footnotesize $B_r$}(7.9,5.95);
\draw [->,line width=0.6pt] (6,4) -- node[above=-0.7ex, rotate= -25] {\footnotesize $B_r$}(7.9,3.05);

%grid
\draw [color=cqcqcq,, xstep=2cm,ystep=1cm] (1,6) grid (8,3);

\draw[fill=rvwvcq] (2,6) circle (2pt);
\draw[fill=rvwvcq] (4,5) circle (2pt);
\draw[fill=rvwvcq] (6,4) circle (2pt);
\draw[fill=rvwvcq] (8,6) circle (2pt);

\end{tikzpicture}
\caption{Events Timings. 
Let replica $i$ be the first honest replica to enter round $r$, at time $\tau_i(r) = \utau(r)$,
upon receipt (or construction) of a notarization $z_{r-1}$ for round $r-1$.
Let $j$ be any other honest replica for illustration.
Replica $i$ broadcasts $z_{r-1}$ to everyone
where it arrives before $\utau(r)+\Delta$.
Thus, all honest replicas %including $j$ 
start their round $r$ before $\utau(r)+\Delta$.
Immediately after starting round $r$,
all honest random beacon members create and broadcast their random beacon share to everyone
where it arrives before $\utau(r) + 2\Delta$.
Immediately after receiving $\xi_r$,
each honest block maker broadcasts its block proposals $B_r$ to everyone 
where it arrives before $\utau(r)+ 3\Delta$.
Replica $i$ finishes its waiting period at $\tau_i(r)+\btime=\utau(r)+3\Delta$
and proceeds to sign the highest priority proposal $B_r$ in its view.
Replica $j$ finishes its waiting period later, at $\tau_j(r)+\btime$,
and does the same. }
\label{fig:events}
\end{figure}
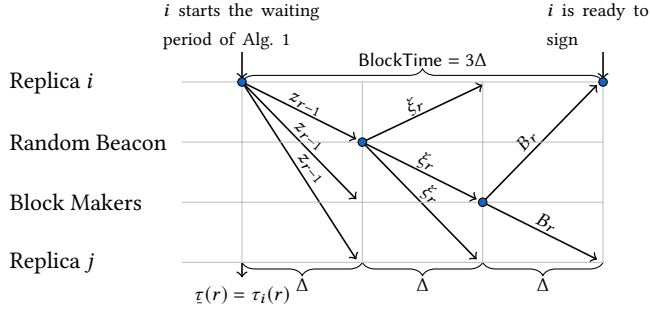

\begin{lemma}[Block, "Slow" Bound]\label{lem:Br}
Suppose $\btime\geq 3\Delta$.
For each round $r$ and each honest block proposal $B_r$ we have:
\begin{align}
\otau(B_r)\leq\ubar{\tau}(r)+\btime.
\end{align}
\end{lemma}
\begin{proof}
Since $B_r$ is proposed by an honest replica $i$,
it is broadcasted immediately when $i$ receives $\xi_r$.
(Note that we generally ignore processing times throughout the section,
including the block creation time.)
This means $\utau(B_r)\leq\otau(\xi_r)$,
hence
$$\utau(B_r)\leq\otau(\xi_r)\stackrel{\eqref{eq:tauxi}}{\leq}\utau(r)+2\Delta\leq\utau(r)+(\btime-\Delta).$$
The assertion follows from Cor.~\ref{cor:fwd} for $A=B_r$.
\end{proof}

\begin{proposition}[Normal Operation]\label{prop:normal}
Suppose $\btime\geq 3\Delta$.
If the highest priority replica in round $r$ is honest,
then round $r$ has normal operation.
\end{proposition}
\begin{proof}
Let $i$ be the highest priority replica of round $r$ 
and suppose $i$ is honest.
Then $i$ proposes exactly one block $B_r$ for round $r$.
Lemma \ref{lem:Br} implies that for each honest notary member $j$
we have $\tau_j(B_r)\leq\tau_j(r)+\btime$.
By Alg.~\ref{alg:psp},
replica $j$ waits $\btime$ after entering round $r$ and then signs $B_r$
and only $B_r$
because $B_r$'s proposer $i$ has the highest possible priority.

Since notarization requires the participation of at least one honest replica,
and all honest replicas sign only $B_r$,
$B_r$ is the only block that can possibly get notarized.
In other words, round $r$ can only be ended by a notarization of $B_r$,
which guarantees that the signatures under $B_r$ saturate the network
and $B_r$ indeed gets notarized.
\end{proof}

\subsubsection{Minimal Progress}
\begin{proposition}[Minimal Progress]\label{prop:progress}
Suppose $\btime\geq 3\Delta$.
Suppose that in a given round $r$ the replica $i$ with $\pi_r(i)=d$ is honest.
Then
\begin{align}\label{eq:minprog}
\utau(r+1)\leq\utau(r)+\btime+(d+2)\Delta.
\end{align}
\end{proposition}
\begin{proof}
Set $x_0:=\utau(r)+\btime$.
Let $B_r$ be the unique proposal made by $i$ for round $r$.
By Lemma \ref{lem:Br},
all honest notary members receive $B_r$ by time $x_0$.

For each time $x\geq x_0$,
we consider the set $S_x$ of valid block proposals for round $r$ 
that have been received by at least one honest replica.
More formally,
$S_x$ consists of those valid round-$r$ proposals $B$ 
that satisfy $\utau(B)\leq x$.
For example, $B_r\in S_{x_0}$.
We then define $$f(x):=\min\{\rank B\,|\, B\in S_x\}.$$
For example,
$f(x_0)\leq d$ because $B_r$ has rank $d$.

For $x\geq x_0$,
the function $f$ has values in the non-negative integers and is monotonically decreasing.
Since $f(x_0)\leq d$,
it follows that there is a time $x_0\leq x_1\leq x_0+d\Delta$
such that $f(x_1)=f(x_1+\Delta)$.
(Otherwise, if $f(x+\Delta)\leq f(x)-1$ for all $x_0\leq x\leq x_0+d\Delta$
then $f(x_0+(d+1)\Delta)<0$, a contradiction.)

We claim $\utau(r+1)\leq x_1+2\Delta$.
This proves \eqref{eq:minprog} since
\begin{align*}
 x_1+2\Delta & \leq x_0+(d+2)\Delta\\
& = \utau(r)+\btime+(d+2)\Delta.
\end{align*}
If a notarization for a block different from $B$ arrives at any replica before $x_1+2\Delta$
then the claim is already proven.
We assume w.l.o.g.\ that this is not the case.
Note that under this assumption,
\eqref{eq:tauijA} applies to the events $B$ and any signature under $B$.

Let $B\in S_{x_1}$ with $\rank B=f(x_1)$.
Since $B\in S_{x_1}$ we have $\utau(B)\leq x_1$.
Hence, by \eqref{eq:tauijA},
$$\otau(B) \leq x_1+\Delta.$$
The facts $\otau(B) \leq x_1+\Delta$ and $\rank B=f(x_1+\Delta)$ together mean
that $B$ has minimal rank in every honest replica's view at time $x_1+\Delta$.
Also, $x_1+\Delta\geq \utau(r)+\btime+\Delta\geq \otau(r)+\btime$,
i.e.\ $x_1+\Delta$ is past the $\btime$ waiting period for every honest replica.
Thus, all honest replicas have broadcasted a signature for $B$ by $x_1+\Delta$.
By \eqref{eq:tauijA},
all honest replicas receive $f+1$ signatures by $x_1+2\Delta$,
i.e.\ $\otau(r+1)\leq x_1+2\Delta$.
\end{proof}
We remark without proof:
In the case $d=0$ the bound can be improved to
$\utau(r+1)\leq\utau(r)+\btime+\Delta$.
The resulting bounds are then strict for all $d$.

As a corollary, by induction, we conclude that the protocol makes continuous progress:
\begin{corollary}\label{cor:progress}
Suppose $\btime\geq 3\Delta$.
For all rounds $r$ and all honest replicas $i$,
$\tau_i(r)$ is finite.
\end{corollary}

\subsection{Near-Instant Finality} \label{sec:instant}
Finality is provided by Alg.~\ref{alg:psp}.
This section first proves the correctness of this algorithm
and then shows as the main theorem that finality is achieved quickly under normal operation.

Recall:
\begin{itemize}
\item
According to the relay policy the notarization will reach all other honest replicas.
\item
As was stated in \eqref{eq:B},
if a block is honestly signed then the notarized previous block will saturate the network.
\item if a block is honestly signed then $\nota B$ is re-broadcasted under the policy Def.~\ref{def:relay}{b)}.
\end{itemize}

\begin{lemma}\label{lem:global}
Suppose $z_{r-1}$ is a referenced notarization for round $r-1$.
Then,
\begin{equation}\label{eq:zwindow}
\otau(z_{r-1})\leq\otau(r+1)+\Delta.
\end{equation}
\end{lemma}
\begin{proof}
Let $B_r$ be a notarized block with $\nota B_r=z_{r-1}$.
Since $B_r$ received an honest signature,
we have $\utau(B_r)\leq\otau(r+1)$.
Since $B_r$ contains $z_{r-1}$, we have $\utau(z_{r-1})\leq\utau(B_r)$.
This implies
$\otau(z_{r-1})\leq\utau(z_{r-1})+\Delta\leq\otau(r+1)+\Delta$. %\leq\utau(r+1)+2\Delta.
\end{proof}

\begin{theorem}[Correctness of Finalization]\label{thm:alg2}
Suppose $T\geq 2\Delta$.
For every honest replica,
before the replica executes $\textsc{Finalize}(h)$ in Algorithm~\ref{alg:chain}
it has received all notarizations for round $h$ that can get referenced.
\end{theorem}
The assertion is precisely the correctness assumption \eqref{eq:correct}.
\begin{proof}
Suppose $z_{r-1}$ is a referenced notarization for round $r-1$.
From \eqref{eq:zwindow} and \eqref{eq:tauij}, we get
\begin{equation}\label{eq:zwindow2a}
 \otau(z_{r-1})\leq\utau(r+1)+2\Delta.
\end{equation}
In particular, for any honest replica $i$,
\begin{equation}\label{eq:zwindow2}
 \tau_i(z_{r-1})\leq\tau_i(r+1)+2\Delta.
\end{equation}
This shows that any notarization for round $r-1$ that arrives at $i$ after $\tau_i(r+1)+2\Delta$ cannot get referenced.
Since Alg.~\ref{alg:chain} calls $\textsc{Finalize}(r-1)$ at time $\tau_i(r+1)+T$
and $T\geq 2\Delta$,
this proves the claim.
\end{proof}
\begin{corollary}\label{cor:final+3}
Suppose $\btime\geq 2\Delta$.
Suppose $z_{r-1}$ is a referenced notarization for round $r-1$.
Then,
\begin{equation}\label{eq:zwindow3}
\otau(z_{r-1})\leq\utau(r+2).
\end{equation}
\end{corollary}
\begin{proof}
By \eqref{eq:taubtime},
$\utau(r+1)+\btime\leq\utau(r+2)$.
Hence, by \eqref{eq:zwindow2a},
$\otau(z_{r-1})\leq\utau(r+1)+2\Delta\leq\utau(r+2)$.
\end{proof}
Provided that $\btime\geq 2\Delta$,
\eqref{eq:zwindow3} provides an alternative criteria for when to execute $\textsc{Finalize}(r-1)$ in Alg.~\ref{alg:chain}
that does not require $T$.
Instead of waiting for $T$ into round $r+1$,
the finalization procedure can simply wait for round $r+1$ to end in the observer's view.

\begin{theorem}[Main Theorem]
Under normal operation in round $r$,
every transaction included in a block for round $r$ is final for an observer
after two confirmations plus the maximum network roundtrip time $2\Delta$.
\end{theorem}
\begin{proof}[Proof of Theorem~\ref{thm:main}]
Suppose round $r$ has normal operation,
i.e.\ only one block $B_r$ gets notarized for round $r$.
We assume the observer has chosen $T=2\Delta$.
At time $T$ after seeing a notarization for round $r+1$,
the observer will finalize round $r$.
Since $\calN_r$ 
(the bucket of all received notarized blocks for round $r$)
contains only $B_{r}$,
$B_r$ is appended to the final chain at this time.
\end{proof}

\subsection{Chain Properties}\label{sec:chainprop}

Recall that each replica at the end of each round has its own view of an append-only finalized chain
(cf.\ Alg.~\ref{alg:chain}).
We consider the following properties regarding state and content of the finalized chain $C$.
\begin{definition}[Chain properties]\mbox{}
\begin{enumerate}[a)]
 \item {\em Growth with parameter $k$.}
 Each honest replica's finalized chain at the end of their round $r$ has length $\geq r-k$.

 \item {\em Consistency.}
 If $C,C'$ are the finalized chains of two honest parties, 
 taken at any point in time,
 then $C$ is a prefix of $C'$ or vice versa.

 \item {\em Quality with parameter $l$ and $\mu$.}
 Out of any $l$ consecutive blocks from the finalized chain of an honest replica,
 at least $\mu l$ blocks were proposed by an honest replica.
\end{enumerate}
\end{definition}

\begin{proposition}
Persistence follows from the properties of chain consistency and chain growth.
\end{proposition}
\begin{proof}
Suppose a transaction is included in the finalized chain of one honest replica $i$,
say in the block at round $r$.
Given any other honest replica $j$, by the growth property, $j$'s finalized chain will eventually reach length $r$ as well.
At that point, the consistency property guarantees that $j$'s block at round $r$ is identical to the one in $i$'s finalized chain.
\end{proof}

\begin{proposition}
Liveness follows from the properties of chain quality and chain growth.
\end{proposition}
\begin{proof}
A transaction originating from an honest replica is picked up by all other honest parties and included in their block proposals.%
\footnote{There may be reasons why a transaction can not be included in a block proposal even if the proposer is honest.
Those reasons, such as limited block space,
are not part of our definition of "liveness"
and are not considered here.
}
The growth property guarantees that the finalized chain will eventually grow arbitrarily long.
The chain quality property applies and guarantees that there will eventually be a block in the finalized chain
that was proposed by an honest replica.
\end{proof}

\subsubsection{Chain Growth}

\begin{proposition}[Chain Growth]\label{prop:growth}
Suppose $\btime\geq 3\Delta$.
Accepting a failure probability of $\rho$,
the Chain Growth property holds with parameter $k=\lceil-\log_\beta\rho\rceil$.
\end{proposition}
\begin{proof}
We first look at the property assuming normal operation in round $r$.
In this case, the chain will be finalized up to and including round $r$ at the end of round $r+1$ plus $2\Delta$
according to Alg.~\ref{alg:chain}.
In terms of round numbers, based on Proposition~\ref{prop:minround}, each round including round $r+2$ takes at least $\btime - \Delta > 2\Delta$, thus at the end of round $r+2$, we can finalize round $r$ so that the finalized chain will have length $r+1$.
This means the property holds with $k=1$.

Whenever the highest priority block maker for $r$ is honest
then round $r$ has normal operation (Prop.~\ref{prop:normal}).
Thus, in general,
normal operation happens with probability at least $1-1/\beta$.
Then, with probability at least $1-(1/\beta)^k$ the property holds with parameter $k$.
Thus, if we equate this probability with the minimal desired success probability of $1-\rho$
and solve for $k$ we get $k=\lceil-\log_\beta\rho\rceil$.
\end{proof}

\subsubsection{Chain Consistency}

\begin{proposition}[Chain Consistency]\label{prop:consistency}
Suppose $T\geq 2\Delta$.
Suppose two honest replicas $i,i'$ output the finalized chains $C,C'$
upon executing $\textsc{Finalize}(h), \textsc{Finalize}(h')$, respectively.
Then $C\leq C'$ or $C'\leq C$.
\end{proposition}
\begin{proof}
Assume w.l.o.g\ $h'\leq h$.
Let $\calN_h,\calN'_h$ be the sets of Algorithm~\ref{alg:chain}
at the time when $\textsc{Finalize}(h)$ is being executed by $i,i'$, respectively.
Since $h'\leq h$, by Prop.~\ref{prop:correct}, $C'\leq C(\calN'_h)$.
Let $X$ be the set of all round-$h$ notarizations that can get referenced.
Note that $X$ is globally defined after all honest replicas have ended their round $h+1$ (though $X$ is not known to anyone).
By Thm.~\ref{thm:alg2} (correctness of finalization),
we have $X\subseteq \calN_h,\calN'_h$.
The set $X$ is furthermore non-empty because $\textsc{Finalize}(h)$ is only called
once some round-$h$ notarizations are actually referenced.
Hence,
$C=C(\calN_h)\leq C(X)$ and 
$C'\leq C(\calN'_h)\leq C(X)$.
The fact that $C,C'$ are prefixes of a common chain proves the claim.
\end{proof}

\subsubsection{Chain Quality}
Whether the highest priority replica in round $r$ is adversarial
can be modelled as a Bernoulli trial $X_r$ with success probability $f(U)/|U|$.
Since the adversary cannot bias the random beacon,
the trials $X_r$ for different $r$ are independent.
Therefore, an execution of the protocol comes with a Bernoulli process $X_1,X_2,\ldots$
that models the success of the adversary in gaining preference on the proposed block.

Following Garay et.al.\ \cite{garay2017bitcoin},
we define an execution of the protocol as {\em typical} if the Bernoulli process
does not deviate too much from its expectation.
For any set of rounds $S$ we define the random variable $X(S):=\sum_{i\in S}X_i$.
\begin{definition}
An execution is {\em $(\epsilon,\eta)$-typical} if,
for any set $S$ of consecutive rounds with $|S|\geq \eta$,
\[
\left| \frac{X(S)-E(X(S))}{E(X(S))}\right| < \epsilon
\]
\end{definition}
The idea is that a) we can guarantee chain quality in all typical executions,
and b) any execution is typical with all but negligible probability.
In practice,
"negligible probability" is defined by the security parameter $\kappa$
and the parameters $\epsilon, \eta$ are a function of $\kappa$.
\begin{proposition}[Chain Quality]\label{prop:quality}
Suppose $\btime\geq 3\Delta$.
In a $(\epsilon,\eta)$-typical execution
the Chain Quality property holds 
for $\mu=(1-1/\beta)(1-\epsilon)$
and $l=\eta$.
\end{proposition}
\begin{proof}
Whenever $X_r=0$
there will be only one notarization for $r$ (Prop.~\ref{prop:normal}).
This notarization will be for the honest proposer's block,
so the honest proposal is guaranteed to be in the finalized chain.
Therefore the chain quality is at least $(1-1/\beta)(1-\epsilon)$ 
for a $(\epsilon,\eta)$-typical execution if $l \geq \eta$.
\end{proof}

\section{Acknowledgements}
We would like to thank
Robert Lauko, Marko Vukolic, Arthur Gervais, Bryan Ford, Ewa Syta, Philipp Jovanovic, Eleftherios Kokoris, Cristina Basescu and Nicolas Gailly
for helpful comments and discussions.

\bibliographystyle{abbrv}
\bibliography{security}

\end{document}